\documentclass[11pt]{article}

\usepackage{amsmath,amsthm,amssymb,amscd,fancybox,epsfig,float}
\usepackage{graphicx}
\usepackage{epstopdf}
\usepackage{multirow}
\usepackage{hhline}
\usepackage{cleveref}
\usepackage{xcolor}
\usepackage{tikz}
\usetikzlibrary{arrows,automata}
\usepackage[nocompress]{cite}
\usepackage{url}
\usepackage{comment}

\usepackage{pgfplots}
\pgfplotsset{compat=1.17}   

\usepackage{adjustbox}
\usetikzlibrary{arrows.meta}

\usepackage{subcaption}
\usepackage{authblk}
\usepackage{blindtext}

\newcommand{\hk}{\hat{k}}         

\newcommand{\inner}[2]{\ensuremath{\langle #1,#2 \rangle}}

\newcommand{\R}{\mathbb{R}}
\newcommand{\C}{\mathbb{C}}
\newcommand{\N}{\mathbb{N}}
\newcommand{\Z}{\mathbb{Z}}

\usepackage{bm}

\newcommand{\tran}{^{\mathsf{T}}}
\usetikzlibrary{positioning}

\newcommand{\bi}{\begin{itemize}}
\newcommand{\ei}{\end{itemize}}
\newcommand{\ben}{\begin{enumerate}}
\newcommand{\een}{\end{enumerate}}
\newcommand{\be}{\begin{equation}}
\newcommand{\ee}{\end{equation}}
\newcommand{\bea}{\begin{eqnarray}} 
\newcommand{\eea}{\end{eqnarray}}
\newcommand{\ba}{\begin{align}} 
\newcommand{\ea}{\end{align}}
\newcommand{\bse}{\begin{subequations}} 
\newcommand{\ese}{\end{subequations}}
\newcommand{\bc}{\begin{center}}
\newcommand{\ec}{\end{center}}
\newcommand{\bfi}{\begin{figure}}
\newcommand{\efi}{\end{figure}}

\newcommand{\bmp}[1]{\begin{minipage}{#1}}
\newcommand{\emp}{\end{minipage}}

\newcommand{\mbf}[1]{{\mathbf #1}}
\newcommand{\half}{\mbox{\small $\frac{1}{2}$}}
\newcommand{\bigO}{{\mathcal O}}
\newcommand{\sfrac}[2]{\mbox{\small $\frac{#1}{#2}$}}

\newcommand{\niter}{{n_{\textrm{iter}}}}
\newcommand{\eps}{\varepsilon}     
\newcommand{\bj}{j}    
\newcommand{\xz}{x}                 
\newcommand{\xt}{x^\ast}                 
\newcommand{\yt}{y^\ast}                 
\newcommand{\kz}{{\mbf{k}_\xz}}       
\newcommand{\tkz}{{\tilde{\mbf{k}}_\xz}}       
\newcommand{\bal}{\bm{\alpha}}       
\newcommand{\tbal}{\tilde{\bm{\alpha}}}       
\newcommand{\bbe}{\bm{\beta}}           
\newcommand{\bmu}{\bm{\mu}}          
\newcommand{\hbmu}{\hat{\bm{\mu}}}          
\newcommand{\X}{\Phi}                
\newcommand{\Xt}{\Phi^*}            
\newcommand{\bpz}{\bm{\phi}_\xz}       

\newtheorem{theorem}{Theorem}
\newtheorem{lem}[theorem]{Lemma}
\newtheorem{algorithm}[theorem]{Algorithm}
\newtheorem{pro}[theorem]{Proposition}

\newtheorem{rmk}[theorem]{Remark}

\title{Equispaced Fourier representations for efficient Gaussian process regression
  from a billion data points}
\author[1]{Philip Greengard\thanks{pg2118@columbia.edu. Research supported by the Alfred P. Sloan Foundation.}}
\author[2]{Manas Rachh}
\author[2]{Alex Barnett}

\affil[1]{Department of Statistics, Columbia University, New York, NY, 10027}
\affil[2]{Center for Computational Mathematics, Flatiron Institute, New York, NY, 10010}

\date{\today}
\usepackage{cleveref}

\begin{document}

\maketitle

\begin{abstract}
We introduce a Fourier-based fast algorithm for Gaussian process regression in low dimensions. It approximates a translationally-invariant covariance kernel by complex exponentials on an equispaced Cartesian frequency grid of $M$ nodes. This results in a weight-space $M\times M$ system matrix with Toeplitz structure, which can thus be applied to a vector in $\bigO(M \log{M})$ operations via the fast Fourier transform (FFT),  independent of the number of data points $N$. The linear system can be set up in $\bigO(N + M \log{M})$ operations using nonuniform FFTs. This enables efficient massive-scale regression via an iterative solver, even for kernels with fat-tailed spectral densities (large $M$). We provide bounds on both kernel approximation and posterior mean errors.  Numerical experiments for squared-exponential and Mat\'ern kernels in one, two and three dimensions often show 1--2 orders of magnitude acceleration over state-of-the-art rank-structured solvers at comparable accuracy. Our method allows 2D Mat\'ern-$\sfrac{3}{2}$ regression from $N=10^9$ data points to be performed in 2 minutes on a standard desktop, with posterior mean accuracy $10^{-3}$. This opens up spatial statistics applications 100 times larger than previously possible.
  \end{abstract}

\section{Introduction}
\label{s:intro}

Gaussian process (GP) regression is ubiquitous in
spatial statistics and machine learning,
due in large part to its high predictive performance
and rigorous foundation in Bayesian inference
\cite{rasmus1, bartok1, stein2, aki1, cressie1,liu20bigreview,murphy23}.
The goal is to recover a stochastic function with a certain covariance kernel,
given noisy pointwise observations.
More precisely, one is given $N$ points
$x_1, \dots, x_N \in \R^d$,
with corresponding data observations $y_1,\dots,y_N$ which are noisy samples from a function $f$ with a zero-mean
Gaussian process distribution. This is denoted by
\begin{align}
y_n & \sim f(x_n) + \epsilon_n ~, \qquad n=1,\ldots,N \\
f(x) & \sim \mathcal{GP}(0, \, k(x, x'))
\end{align}
where $\epsilon_n \sim \mathcal{N}(0, \sigma^2)$ is independent and identically
distributed (iid) noise of known variance $\sigma^2>0$,
and $k: \R^d \times \R^d \to \R$ is a given
covariance kernel.
We address the common {\em translationally invariant} case
where $k(x,x') = k(x-x')$, using these notations interchangeably.
This covers a wide range of statistical problems, where
$k$ is generally chosen to belong to one of a handful of 
families of kernels such as Mat\'ern, squared-exponential (SE),
or rational quadratic.
In this paper we present a regression method efficient in low
dimensions $d$ (say $d\le 5$), targeting
applications in time series ($d=1$) and spatial/geostatistics ($d=2$ or $3$);
we do not address high-dimensional machine learning tasks.

In the GP prior model the vector $\{f(x_n)\}_{n=1}^N$ is
a zero-mean multivariate Gaussian with $N\times N$ covariance
matrix $K$ with elements $K_{nl} := k(x_n,x_l)$.
Then, given data $\mbf{y} := \{y_n\}_{n=1}^N$, the regression task
to compute the (necessarily Gaussian) posterior of $f(\xz)$,
at new target points $\xz\in\R^d$.
A standard calculation \cite[\S 2.2]{rasmus1}
(conditioning a $N+1$ dimensional multivariate
Gaussian on the known subset of $N$ variables)
gives the posterior mean of $f(\xz)$ as
\bea
\mu(\xz) &=& \sum_{n=1}^{N} \alpha_n k(\xz,x_n),
\label{mu}                 
\eea
where $\bal := \{\alpha_n\}_{n=1}^N$ is the vector in $\R^N$
uniquely solving the $N\times N$ symmetric linear system
\bea
(K + \sigma^2 I) \bal &=& \mbf{y},
\label{362}
\eea
$I$ being the $N\times N$ identity matrix.
Computing $\mu(x)$, the main goal of this paper,
is equivalent to {\em kernel ridge regression}
\cite[Ch.~18.3.7]{murphy23}.
\begin{rmk}[variance evaluation]\label{r:var} 
  Formulae similar to \eqref{mu}--\eqref{362}
  exist for the variance $s(x)$ (or covariance between targets)
  \cite[(2.26)]{rasmus1} \cite[\S1--2]{efgp_combo}.
  Unlike the mean, they demand a new linear solve for each target.
  Our proposed accelerated method would also apply to such cases,
  with scaling as in the last column of Table~\ref{t:comps};
  for space reasons we leave this for future study.
\end{rmk}

The main obstacle to the use of GP regression as a large-scale practical tool
is the solution of linear systems such as \eqref{362}.
Dense direct solution (e.g.\ by Cholesky factorization
of $K+\sigma^2 I$)
costs $O(N^3)$ operations, limiting $N$ to of order $10^4$
on a single machine, prohibitive for even a moderately
small modern data set.
Thus much literature has emerged on more scalable methods.
While most of these tackle the above $N\times N$ linear system
(this is called the ``function space'' approach),
others solve a different 
linear system for basis coefficients
(``weight space'') \cite[\S2]{rasmus1} \cite[Ch.~18]{murphy23}.
\begin{table}
\centering
\resizebox{\linewidth}{!}{
\renewcommand{\arraystretch}{2}
\begin{tabular}{c | c |c|c|c|c}
  method
  & error (SE kernel)$^{\ast}$ & precomputation & linear solve & mean at $q$ points & variance at $q$ points \\
  \hline
  direct & $\epsilon_{\text{mach}}$ & $N^2$ & $N^3$ & $Nq$ & $N^2 q$ \\
  \hline
  SKI \cite{wilson2} & $M^{-3/d}$ & -- & $n_{\text{iter}} (N + M \log{M})$ & $Nq$ & $n_{\text{iter}} (N + M \log M) q$ \\
  \hline
  FLAM \cite{minden1} & set by $\eps$ & $\max[N\log N, N^{3-3/d}]$ & $N \log{N}$ & $(N + q) \log(N + q)$ & $Nq \log{N}$ \\
  \hline
  RLCM \cite{chen1} & set by rank & $N$ & $N$ &
  $N + q\log{N}$ & $N + q\log{N}$ \\  
  \hline 
  PG '21 \cite{greengard3} & $e^{- \gamma M^{1/d}}$ & $N + M^2 \log{M}$ & $M^3$ & $q + M\log{M}$ & $M^2q$ \\
  \hline
  EFGP & $e^{-\gamma M^{1/d}}$ & $N + M \log M$ & $n_{\text{iter}} \, M \log{M}$ & $q + M\log M$ & $n_{\text{iter}}  Mq \log{M} $
\end{tabular}
}
 \caption{Error and computational complexity of GP regression
   for different algorithms, for $N$ points defined on $\R^d$, with regression to
   $q$ new targets.
   All entries are $\bigO(\cdot)$; constants are not compared.
 For ``SKI,'' $M$ denotes the number of inducing points;
 otherwise $M$ is the number of Fourier nodes.
 For ``EFGP,'' the proposal in this paper, $M=(2m+1)^d$, since $2m+1$ nodes are used in each dimension.
 $n_\text{iter}$ is the number of iterations in an iterative solver,
 $\eps$ is a matrix compression tolerance,
 and $\epsilon_\text{mach}$ is machine error.
 ``Error'' refers to kernel approximation,
 not downstream prediction errors that also scale with $N$ and $\sigma$.
 $^\ast$SE denotes squared-exponential.
}
\label{t:comps}
\end{table}

\subsection{Prior work in GP regression in low dimensions}
Much prior work has applied
{\em iterative solvers}, usually conjugate gradients (CG)
\cite[\S10.2]{golubvanloan}, to the $N\times N$
function-space linear system \eqref{362}. 
This requires multiplications of arbitrary vectors by
the matrix $K+ \sigma^2 I$.
These may be performed densely with $\bigO(N^2)$ cost per iteration
\cite{gibbsmackay97};
for a brute-force parallel implementation
see \cite{wang19exact}.
However, most approaches instead apply an approximation to $K$
which exploits the {\em low-rank structure} of covariance matrices.%
\footnote{These are sometimes called ``sparse'' GPs \cite{hensman1,liu20bigreview}.}
The kernel approximation $k(x-x') \approx \tilde k(x-x') :=
\sum_{j=1}^M \phi_j(x)\phi_j(x')$, for a set of $M$ basis functions $\phi_j$,
leads to a global factorization
$K \approx \tilde K := \X \Xt$, whose rank is thus at most $M$.
Perhaps the most popular examples of this are
the Nystr\"om method \cite{nystrom,rasmus1},
inducing point methods \cite{RRGP05,snelson05},
and subset of regressors (SoR) \cite[\S8.3.1]{rasmus1}.
These factorizations often involve a subsampling of rows and columns
of $K$, so that the bases are themselves kernels centered at a size-$M$
subset of the given points.
Their cost is typically $\bigO(N M^2)$.
Building on inducing points, ``structured kernel interpolation'' (SKI)
\cite{wilson2} (similar to precorrected FFT in PDE solvers \cite{precorrfft})
replaces these $x_{n_j}$ by a {\em regular spatial product grid} of
interpolation points;
the resulting Toeplitz
structure allows a fast matrix-vector
multiply via the $d$-dimensional FFT,
giving the much improved scaling $\bigO(N + M\log M)$ per iteration
(see Table~\ref{t:comps}).
Recently, analytic expansions have improved upon its speed at high accuracy
up to moderate $d$ \cite{FKT22}.

Most such iterative methods fit into the framework of
``black box'' matrix-vector multiplications (BBMM) \cite{gardner1} with
$K$.
However, their disadvantage is that each iteration costs at least $\bigO(N)$
while poor condition numbers mean that many iterations are required.
Preconditioning can help in certain settings
\cite{steinprecond,gardner1}, but is far from a panacea
for GP regression \cite{cunn08toep}.
Furthermore, the cost per
posterior mean evaluation at each new target is again $\bigO(N)$.
In 2015 it was stated that ``with sparse approximations it is inconceivable to apply GPs to training set sizes of $N \ge \bigO(10^7)$'' \cite{distribgp}.
The largest data sizes tested
on a desktop machine to date in this framework are around $N\approx 10^6$
\cite{gardner1,wang19exact,liu20bigreview}.

We now turn to prior work in {\em direct solvers},
which have the advantage of immunity to ill-conditioning.
Simple truncation of the kernel beyond a certain range
sparsifies $K$, allowing a sparse direct solution \cite{furrertaper},
but this can only be effective for short-range kernels.
In dimension $d=1$ (e.g.\ time series),
certain kernels make the upper and lower triangular
halves of $K$ low rank; this
``semiseparable'' structure allows a direct solution in $\bigO(N)$ time
\cite{dfm1}.
Recently, fast direct solvers have
been successfully applied to the function-space linear system
in low dimension $d>1$ \cite{oneil1,minden1,stein1,chen1,kesha22}.
Going beyond the above ``one-level'' compression for $K$,
these exploit hierarchical (multi-scale) low-rank structure
to compute a compressed approximate representation of $(K+\sigma^2 I)^{-1}$,
which may be quickly (in $\bigO(N)$ or $\bigO(N \log N)$ time) applied
to new vectors $\mbf{y}$ in \eqref{362}.
For instance, in HODLR (hierarchically off-diagonal low-rank) structure,
at each scale off-diagonal blocks are
numerically low rank, while the remaining diagonal blocks themselves
have HODLR structure at a finer scale \cite{oneil1}.
Refinements of this structure (e.g., HSS, $\cal H$, ${\cal H}^2$, and HBS),
as well as randomized linear algebra and ``proxy point'' sampling
\cite{minden1,gunnarbook}, are needed to
achieve low ranks and fast factorizations.
These build upon recent advances in fast direct solvers for various
kernels arising in integral equations and PDEs \cite{hackbusch,gunnarbook},
which in turn exploit adaptive domain decomposition and far-field expansion
ideas from tree codes and the fast multipole method \cite{lapFMM}.
One disadvantage of fast direct solvers
is that the ranks grow rapidly with
$d$, and that compression times scale quadratically with the ranks
\cite{oneil1}.
Another is their high RAM usage.
The largest $N$ is around $2\times 10^6$ in recent single workstation
examples \cite{oneil1,minden1,chen1}.

In contrast to all of the above function-space methods,
{\em weight-space methods} avoid solving the $N\times N$ linear system
\cite{RRGP05}.
Recalling the low-rank kernel approximation $K \approx \tilde K = \X \Xt$,
they instead solve a dual system involving the smaller
$M\times M$ matrix $\Xt\X$.
We recap the equivalence of these two approaches in Section~\ref{s:fac}.
The needed rank (basis size)
$M$ depends on the kernel, dimension $d$, and desired accuracy.
An advantage is that mean prediction at new targets
involves only evaluating the basis expansion, an $\bigO(M)$ cost
{\em independent of the data size} $N$.
However, traditional weight-space methods (including SoR
and ``sparse''/inducing variable methods \cite{RRGP05,snelson05})
cost $\bigO(NM^2)$ to form $\Xt\X$, again prohibitive for large $N$.

Recently it has been realized that Fourier
(real trigonometric or complex exponential) basis functions $\phi_j(x)$ are
a natural weight-space representation for translationally invariant
kernels (for a review of such spectral representations see \cite{hensman1}).
In early work their frequencies
were chosen stochastically \cite{rahimi1},
or optimized as part of the inference \cite{lazaro1}.
Instead, Hensman et al.\ \cite{hensman1}
proposed regular (equispaced) frequencies,
following \cite{paciorek1}. 
This allowed regression
(in a separable model, i.e., effectively $d=1$) from $N=6\times 10^6$ points
in around 4 minutes.  
However, since direct Fourier summation was used, the cost was still
$\bigO(NM)$.

Recently the first author introduced
a weight space method for $d=1$
with Fourier frequencies chosen at the nodes of a
{\em high-order accurate quadrature scheme} for the Fourier integral~\cite{greengard3}.
Since smooth kernels (such as SE or Mat\'ern with larger $\nu$)
have rapidly decaying Fourier transforms, truncation to
a small $M$ can approximate $\tilde K \approx K$ to many digits of accuracy,
giving essentially ``exact'' inference.   
That work also proposed the nonuniform FFT (NUFFT \cite{dutt1}) of type 3
for filling of the weight space system matrix and evaluation at targets,
leading to a solution cost $\bigO(N + M^3)$; see Table~\ref{t:comps}.
This allowed $N=10^8$ to be handled in a few seconds.

Finally, for $d>1$, the only way to handle
$N \gg 10^6$
currently is via distributed computing \cite{liu20bigreview}.
There, methods include approximation by mixtures (a ``committee'' \cite[\S8.3.5]{rasmus1}) of
fixed-size GP solutions \cite{distribgp},
distributed gradient descent \cite{peng17billion}, and multi-CPU Cholesky
preconditioning \cite{meanti20}.
For instance, in
studies with $N=10^9$ points (in $d=9$), convergence took around 1 hour,
either on 1000 CPU cores \cite[Sec.~6.3]{peng17billion}, or 2 GPUs \cite{meanti20}.

\subsection{Our contribution}
\label{s:contrib}

We present a weight space iterative solver for general low dimension
using Fourier frequencies (features) on an equispaced $d$-dimensional grid,
as in \cite{hensman1}.
The equal spacing makes the $M$-by-$M$ system matrix
Toeplitz (or its generalization for $d>1$)
up to diagonal scalings, enabling
its multiplication with a vector via the padded FFT, to give
$\bigO(M\log M)$ cost per iteration.
(Although SKI \cite{wilson2} also uses an FFT, the 
grid there is in physical space, forcing an $\bigO(N)$ cost per iteration
as with all BBMM methods.)
The result is that large $M$ (e.g.\ $10^6$)
may be tackled (unlike in \cite{hensman1,greengard3}).
This is a necessity for higher dimension
since $M = \bigO(m^d)$, where $m$ controls
the number of grid points per dimension.
Furthermore, building on \cite{greengard3},
we fill the vector defining the Toeplitz matrix, and the right hand
side, 
via a modern implementation of the NUFFT \cite{barnett1}
with $\bigO(N + M\log M)$ cost.
Since there is only a {\em single pass through the data},
a standard desktop can handle very large $N$,
around 100 times larger than any prior work
of which we are aware in $d=2$.

We benchmark our equispaced Fourier (EFGP) proposal against
three popular methods for GP regression that have
desirable theoretical properties
and 
available software:
SKI of Wilson et al.\ \cite{wilson2} implemented
in GPyTorch \cite{gardner1},
an implementation of the fast direct solver
of Minden et al.\ \cite{minden1} using FLAM \cite{ho1},
and the linear-scaling direct solver
RLCM of Chen and Stein \cite{chen1}.
Their theoretical scalings are compared in Table~\ref{t:comps}.
Our synthetic data tests in dimensions $d=1$, $2$, and $3$
explore the tradeoff of accuracy vs CPU time for all four methods,
and show that EFGP is very competitive
for the very smooth squared-exponential kernel,
and remains competitive even for
the least smooth Mat\'ern kernel ($\nu=1/2$).
We show that real satellite data in $d=2$ with $N=1.4\times 10^6$
may be regressed in a couple of seconds on a laptop.
We show large-scale simulated data tests in $d=2$ at up to a billion
data points; the latter takes only a couple of minutes on a desktop.
For reproducibility,
we implemented or wrapped all methods
in a common, documented, publicly available MATLAB/Octave interface.

Our contribution includes both analysis and numerical study of the
error in the computed posterior mean $\mu$ relative to 
``exact'' GP inference (exact solution of \eqref{mu}--\eqref{362}).
If this error---the bias in the GP regression introduced by approximation---%
is pushed to a negligible value then
the computation becomes effectively exact
(a point appreciated recently with other iterative methods
\cite{gardner1,wang19exact}).
We bound this error rigorously in our main Theorem~\ref{t:muerr}
in terms of $\eps$, a uniform kernel approximation error,
and $\rho$, the weight space residual norm arising from the iterative
solver.
We show that with our EFGP representation $\eps$ splits into
aliasing and truncation parts (Proposition~\ref{p:pw_err}).
For the popular squared-exponential and Mat\'ern kernels,
we give formulae for numerical discretization parameters that
guarantee a given $\eps$, using analytic results from \cite{efgp_anal}.
Numerically we study the
{\em estimated error in posterior mean} (EEPM),
which should converge to zero as $\eps$ and $\rho$ vanish.
This is more informative than the
prediction error at held-out targets (commonly RMSE),
which instead bottoms out at the noise level ($\sigma$
in the case of a well-specified model).

\begin{rmk}[Connection to Fourier imaging]
  \label{r:imaging}
  Although the combination of equispaced Fourier modes with an iterative FFT-accelerated Toeplitz solver appears to be new in GP regression,
  a similar idea has been used in
  computed tomography \cite{delaney1996}, MRI \cite{fesslertoep,lgreengard2006},
  and cryo-electron microscopy \cite{wang1}.
  The normal equations in those applications
  are mathematically equivalent to the GP weight-space linear system with $\sigma=0$, but, curiously, the roles of physical and Fourier space are swapped.
\end{rmk}

\subsection{Outline of this paper}

In the next section we present the weight-space approach
for an exactly factorized covariance kernel, and review its equivalence
to GP regression in function space.
In Section~\ref{s5} we describe the proposed equispaced Fourier
(EFGP) method and its fast numerical implementation.
Our error analysis comprises
kernel approximation (Section~\ref{s:discr}),
discretization parameter choice (Section~\ref{s:choice}),
and posterior mean error (Section~\ref{s:meanacc}).
Section \ref{s:numerics} presents
numerical experiments with EFGP, and benchmarks it against
three state of the art GP regression methods.
We conclude with a discussion in Section~\ref{sec:conclusion}.

\section{Exact kernel factorization and equivalence of weight space}
\label{s:fac}

The function-space formula \eqref{mu}--\eqref{362}
involve solving for coefficient vectors in $\R^N$.
Our proposal instead solves for weight space coefficient vectors in $\C^M$.
The goal of this section is to recap their equivalence
in the case of a kernel possessing an exact factorization
(a degenerate kernel).
Although for real-valued basis functions this is implicit
in literature via the matrix inversion lemma (e.g.\ \cite{RRGP05} \cite[p.~12]{rasmus1}),
a self-contained proof for the complex-valued case is useful.
See \cite{efgp_combo} for weight-space equivalence for the variance $s(x)$.

We assume that data points $x_n$ and target $\xz$ lie in the domain $D := [0,1]^d$;
this can always be achieved by a translation and scaling of the problem.
Suppose now that the kernel $k$ function
admits a rank-$M$ factorization of the form
\begin{align}\label{363}
  k(x, x') = \sum_{j=1}^M \phi_j(x)\overline{\phi_j(x')}, \qquad \forall x,x' \in D,
\end{align}
for a given set of 
basis functions $\phi_j : D \to \C$.
Let $\X$ be the $M\times N$ matrix%
\footnote{This is commonly called the ``design'' matrix, although our definition is
the transpose of that in \cite{rasmus1}.}
with elements $\X_{nj} = \phi_j(x_n)$.
Then \eqref{363} implies the matrix factorization
\be
K= \X\Xt,
\label{Kfac}
\ee
where $^*$ indicates conjugate transpose.
While the covariance kernels
of interest, such as the SE and Mat\'ern kernels, do not admit 
such an exact factorization,
we will shortly see that it is possible to
approximate them to high accuracy by one
\cite{wathenzhu15,gardner1,greengard3}.

If $M<N$, or if fast solvers are available in weight space,
it becomes more efficient to predict the posterior mean by solving a linear system for
a weight vector $\bbe$
as follows, instead of the
function-space vector $\bal$ from \eqref{362}.
\begin{lem}[Equivalence]
  \label{l:equiv} 
  Let $f$ be a GP, with $\{x_n\}_{n=1}^N$ and $\mbf{y}$
  as in the Introduction, and kernel $k$ factorized as \eqref{363}.
  Then the posterior mean \eqref{mu} of $f(\xz)$ may also be written
\be
\mu(\xz) \;=\;\sum_{j=1}^M \beta_j \phi_j(\xz)
\label{muws}
\ee
where $\bbe := \{\beta_j\}_{j=1}^M\in\C^M$
uniquely solves the $M\times M$ linear system
\be
(\Xt\X + \sigma^2 I) \bbe \;=\; \Xt \mbf{y}.
\label{be}
\ee
\end{lem}

To prove this we use the following useful relations
between function and weight space solution vectors;
they may also (less directly) be derived from the Woodbury formula \cite[\S2.1.3]{golubvanloan} \cite[p.~171]{rasmus1}.
\begin{pro}
  Let $\bal\in\R^N$ solve the linear system \eqref{362} and $\bbe\in\C^M$ solve
  the linear system \eqref{be}.
  Then
  \bea
  \bbe &=& \Xt \bal, \label{a2b} \\
  \bal &=& \frac{\mbf{y} - \X \bbe}{\sigma^2}.    \label{magic}
  \eea
\end{pro}
\begin{proof}
Since $\sigma^2>0$, both systems are symmetric positive definite,
so have unique solutions.
  Left multiplying \eqref{362} by $\Xt$ and using $K=\X\Xt$ gives
      $(\Xt\X + \sigma^2 I)\Xt \bal = \Xt \mbf{y}$ which shows that $\Xt \bal$
  solves \eqref{be}, so equals $\bbe$ by uniqueness.
  \eqref{magic} holds if and only if
  $(\mbf{y}-\X\bbe)/\sigma^2$ satisfies \eqref{362},
  i.e., $(K+\sigma^2 I)(\mbf{y}-\X\bbe) = \sigma^2 \mbf{y}$.
  This is easily checked by expanding the left side,
  then substituting $K\mbf{y} = (K+\sigma^2I)\X\bbe$
  which follows by left-multiplying \eqref{be} by $\X$.
\end{proof}

\begin{proof}[Proof of Lemma~\ref{l:equiv}]
    Applying the factorization \eqref{363} to \eqref{mu} gives
    $$
    \mu(x) = \sum_{n=1}^N \alpha_n k(x,x_n)
    =\sum_{n=1}^N \alpha_n\sum_{j=1}^M \overline{\phi_j(x_n)}\phi_j(x)
    =\sum_{j=1}^M (\Xt \bal)_j\phi_j(x),
    $$
    then using \eqref{a2b} proves \eqref{muws}.
\end{proof}

A crucial advantage of solving weight space linear systems
is that, once the matrix $\Xt\X$ has been constructed, the cost
of solution and prediction becomes independent of the training data size $N$.
We will see that by using a fast algorithm
for the sum \eqref{muws},
the cost for mean prediction at $q$ points is linear in $q$ and nearly linear in $M$
(see Table~\ref{t:comps}).

\section{Equispaced Fourier representation of GPs}
\label{s5}

Here we present a numerical approach for constructing a low-rank
{\em approximation} (of the form \eqref{363})
to a translationally invariant covariance kernel $k(x - x')$. 
We use complex exponentials with an equispaced grid of frequencies
as the basis functions $\phi_j$.
For $d>1$ a tensor-product grid is needed, causing $M$ to grow
exponentially with $d$ and limiting its efficient use to, say, $d\lesssim 5$.

\subsection{Discretized inverse Fourier transform}\label{s15}

Let us assume that $k(x-x')$ is 
an integrable and translation-invariant positive covariance kernel.
For Fourier transforms we use the convention of \cite{rasmus1},
\begin{align}
  & \hat{k}(\xi) = \int_{\R^{d}} k(x) e^{-2\pi i \inner{\xi}{x}}\,  dx,
\qquad   \xi \in \R^{d},
  \\
  & k(x) = \int_{\R^{d}} \hat{k}(\xi) e^{2\pi i \inner{\xi}{x}} \, d\xi,
\qquad x \in \R^{d},
  \label{10}
\end{align}
where $\inner{\xi}{x} := \sum_{l=1}^d \xi^{(l)} x^{(l)}$ indicates inner (dot) product.
The basis function approximations we use for Gaussian process distributions 
can be viewed as discretized versions of the Fourier inversion formula \eqref{10}.
In particular, given some constant grid spacing $h>0$, we approximate $k$ via
\begin{align}\label{110}
  k(x-x') \;\approx\;
  \tilde{k}(x - x') = \sum_{\bj \in J_{m}} h^d \hk(h\bj) e^{2\pi i h \inner{\bj}{x - x'}}
  ,
  \qquad x-x' \in [-1,1]^d,
\end{align}
where the multiindex $\bj:=(j^{(1)},j^{(2)},\dots,j^{(d)})$
has elements $j^{(l)} \in \{-m,-m+1,\dots,m\}$, thus
ranges over the tensor product set
$$
J_{m} := \{-m,-m+1,\dots,m\}^d
$$
which contains $M=(2m+1)^d$ elements.
This is simply the {\em trapezoidal rule}%
\footnote{Note that the usual 1D trapezoid rule definition has factors of $1/2$ applied to the
first and last points \cite[p.~209]{na}; this difference is inconsequential
in our application since errors will instead be dominated by truncation to the box.}
applied to each dimension of the
integral \eqref{10} truncated to the box $[-mh,mh]^{d}$.
It is then easy to check that the choice of complex-valued basis functions
\be
\phi_\bj(x) := \sqrt{h^d\hk(h\bj)} e^{2\pi i h \inner{\bj}{x}}, \qquad \bj\in J_m,
\label{basis}
\ee
allows \eqref{110} to be written
as the factorization
$\tilde{k}(x-x') = \sum_{\bj \in J_{m}} \phi_{\bj}(x) \overline{\phi_{\bj}(x')}$,
analogous to \eqref{363}.
The above square-roots are all real
because $k$ is positive:
Bochner's theorem (e.g., \cite[Thm.~4.1]{rasmus1}) guarantees that
$\hk(\xi)\ge 0$, $\forall\xi\in\R^d$.
We call \eqref{basis} the {\em equispaced Fourier} basis.

There are two primary reasons that the approximation of the covariance kernel by using the
truncated trapezoid rule for its Fourier transform is attractive 
in this environment. 
Firstly, this rule is very accurate for smooth functions
that vanish sufficiently fast away from the origin
\cite{weidemanrev}.
For instance, given a smooth ($C^{\infty}$) 1D integrand whose value and all derivatives
vanish at the endpoints,
the error of the $2m+1$ point trapezoid rule is faster than 
$\bigO(1/m^p)$ for any order $p$;
this is called super-algebraic convergence \cite[Thm.~9.26]{na}.
Many commonly-used covariance kernels have smooth Fourier
transforms that vanish rapidly with respect to frequency magnitude $\|\xi\|$.
However, certain rougher kernels such as the Mat\'ern with $\nu=1/2$ \cite{rasmus1}
do not,
so the error in approximating their Fourier inversion
will be dominated by truncation.
Secondly, equispaced Fourier discretizations facilitate the use of
FFTs and NUFFTs, as we describe next.

\subsection{Numerical implementation}
\label{s25}

Here we detail the application of fast Fourier methods to
GP regression using an approximated kernel $\tilde{k}$
factorized as in \eqref{110}.
To predict posterior mean, the task is to form and solve the
weight-space linear system \eqref{be} for the vector $\bbe\in\C^M$,
recalling that the $M=(2m+1)^d$ basis functions are given by \eqref{basis}.
The system matrix is $\Xt\X + \sigma^2I$,
the $N \times M$ design matrix $\X$ being
\bea\label{205}
\X_{n\ell} = \sqrt{h^d\hk(h j_{\ell})}  e^{2\pi i h \inner{j_{\ell}}{x_n}},
\eea
where $x_1,\dots,x_N$ are the observation points, $j_{1},\dots,j_{M}$ is some
enumeration (ordering) of the multiindex
set $J_{m}$, and $h$ is the fixed frequency grid spacing.
The $p$th column of $\X$ is the $p$th complex exponential (with wave vector $hj_p\in\R^d$)
tabulated at all training data points; i.e., $\X_{n,p} = \phi_{j_p}(x_n)$.
Now observe that $\X=FD$, where $F$ is the $N \times M$ matrix
\begin{equation}
F_{n\ell} = e^{2\pi i h \inner{j_{\ell}}{x_n}},
\end{equation}
and 
${D}$ is the diagonal $M \times M$ matrix with
\begin{equation}\label{117}
D_{\ell \ell} = \sqrt{h^{d} \hk(h j_{\ell})}.
\end{equation}
Thus we can factor the system matrix as
\begin{align}\label{97}
  \Xt\X + \sigma^2 I \;=\; D (F^* F) D + \sigma^2 I.
\end{align}

Writing out the matrix product for 
the above square matrix $F^* F$ gives its elements
\begin{align}\label{223}
  ({F}^* {F})_{p, \ell} = \sum_{n=1}^{N} e^{2\pi i\, \inner{h(j_p - j_{\ell})\,}{\,x_n}},
  \qquad p,\ell =1,\dots, M,
\end{align}
which we note only depend on the wave vector {\em difference} $h(j_p - j_{\ell})$.
Consider the 1D case:
since the frequency grid is equispaced,
$j_p-j_\ell$ may only take the $4m+1$ distinct values in $\{-2m,-2m+1,\dots,2m\} = J_{2m}$.
The result is that $F^*F$ has Toeplitz structure (constant along each diagonal).
Generalizing to $d\ge1$,
the set of unique elements is still $J_{2m}$,
and $F^*F$ has $d$-dimensional Toeplitz structure.
For example, in $d=2$ this is called block Toeplitz with Toeplitz blocks (BTTB \cite[Ch.~5]{toepiterbook}).
Note that there is no additional Kronecker structure
\cite[\S2.3.1]{wilson2}
even if the kernel is separable (as is, e.g., the SE kernel);
this is because the data points $\{x_n\}$ in \eqref{223}
are not in general separable.
For any $d$, this Toeplitz structure means that the action of $F^*F$ on a vector,
call it $\bbe$, is
a non-periodic discrete $d$-dimensional convolution
by the vector $\mbf{v}\in \C^{(4m+1)^d}$ with elements
\begin{align}\label{88}
v_{\bj} = \sum_{n=1}^{N} e^{2\pi i h \inner{\bj}{x_n}} , \qquad \bj \in J_{2m}.
\end{align}
To complete the action,
the convolution output must then be
restricted to the central set elements $J_m \subset J_{2m}$.
Yet, this means that it can be evaluated by a similar restriction of the {\em periodic}
convolution with period at least $4m+1$ in each dimension.
This can in turn be performed in near-optimal time via FFTs
\cite[p.~12 and A.7]{toepiterbook}.
The precise method is as follows:
\begin{algorithm}[Fast product of Toeplitz matrix $F^*F$ with a vector $\bbe\in\C^M$]\label{a:FF}
  \mbox{}   
\ben
\item {[precomputation]} Fill $\mbf{v}$ as in \eqref{88},
  which may be done in $\bigO(N+M\log M)$ effort via a $d$-dimensional type 1 nonuniform FFT (NUFFT) \cite{dutt1,lgreengard1,barnett1} with unit strengths.
\item {[precomputation]} Take the $d$-dimensional size-$(4m+1)^d$ FFT of $\mbf{v}$ to get $\hat{\mbf{v}}$.
\item  Zero-pad $\bbe$ in each dimension 
  to become size $(4m+1)^d$, then take 
  its $d$-dimensional 
  FFT.
\item Pointwise multiply the previous result by $\hat{\mbf{v}}$.
\item Take its size-$(4m+1)^d$ inverse FFT.
\item Extract the sub-cube of $M=(2m+1)^d$ outputs
  (if the padding is done as above, these have indices
  $\{2m+1,2m+2,\dots,4m+1\}$ in each dimension, in 1-based indexing.)
\een
\end{algorithm}
Here the precomputation scales linearly with $N$,
and need be done only once given the set $\{x_n\}_{n=1}^N$ of training data points.
After that, each product takes $\bigO(M \log M)$ effort.

Returning to the system matrix \eqref{97},
we can now apply $\Xt\X$ in three steps:
i) diagonal multiplication by $D$, ii) multiplication by the
matrix $F^* F$ by Algorithm~\ref{a:FF}, finally iii) diagonal multiplication by $D$.
In an iterative solver such as conjugate gradients,
the effort per iteration to apply the system matrix
is thus $\bigO(M \log M)$,
with dominant cost per iteration being two FFTs each of size approximately $2^dM$.
In particular, the iteration cost is independent of $N$.

Finally we need a recipe for filling the right-hand side
for the linear system \eqref{be}.
This is the vector $\Xt\mbf{y} = D F^* \mbf{y}$,
where the vector with elements
\be
(F^*\mbf{y})_{\bj} := \sum_{n=1}^N e^{2\pi i \inner{h\bj}{x_n}} y_n,
\qquad \bj\in J_m,
\label{rhs}
\ee
may be computed by a $d$-dimensional type 1 NUFFT with $\bigO(N + M \log M)$
effort. The multiplication by $D$ is a negligible $\bigO(M)$.

By combining the above, the linear system \eqref{be} may be solved
with iterative methods
such as conjugate gradient \cite{dahlquist1}, provided that the matrix 
is not too poorly conditioned,
with $\bigO(M \log M)$ cost per iteration, after a one-time
$\bigO(N + M \log M)$ precomputation.
This idea is the core of the following full EFGP algorithm.
\begin{algorithm}[GP posterior mean using equispaced Fourier representation]\label{a1}
  \mbox{}   
  \begin{enumerate}
\item 
Given a tolerance $\varepsilon$,
determine $m,h$ such that $\tilde{k}$ is an $\varepsilon$-accurate
approximation of the kernel $k$.
We postpone the estimation of $m,h$ that satisfies this criterion
to Section~\ref{s6}.

\item
  Perform the precomputations from Algorithm~\ref{a:FF} above
  to fill $\hat{\mbf{v}}$, where $\mbf{v}$ is defined by \eqref{88}.
  This is dominated by a $d$-dimensional type 1 NUFFT of cost
  $\bigO(N + M \log M)$.
  
\item 
  Compute the right-hand side vector $\Xt \mbf{y}$ using \eqref{rhs}, dominated by
  another $d$-dimensional type 1 NUFFT of cost
  $\bigO(N + M \log M)$.
  
\item 
  Use conjugate gradient to solve the linear system \eqref{be}
  for the vector $\bbe$,
  stopping at relative residual $\rho\approx\varepsilon$.
  At each iteration,
  Algorithm~\ref{a:FF} and diagonal multiplications as above are used to apply the system
  matrix $\Xt\X + \sigma^2 I$ to a vector with cost $\bigO(M \log M)$.

\item
  For the mean $\mu(\xz)$ at a single new target $\xz\in D$,
  simply evaluate \eqref{muws}.

\item
  For the mean at a large set of new targets $z_1,\dots,z_q\in D$,
  use a single $d$-dimensional type 2 NUFFT
  with coefficients $\sqrt{h^d \hk(h\bj)} \beta_\bj$
  to evaluate all Fourier sums \eqref{muws} with $\bigO(q + M \log M)$ cost.
  
\item
  For the mean vector
  at the observation points, $\bmu := \{\mu(x_n)\}_{n=1}^N$, 
  use a single $d$-dimensional type 2 NUFFT
  to evaluate
  \be
  \bmu = \X \bbe,
  \label{bmu}
  \ee
  which is a special case of the previous step with $q=N$.
\end{enumerate}
\end{algorithm}

Here the mean formula \eqref{bmu} follows either from \eqref{muws},
or by inserting \eqref{magic} into $\bmu = K\bal$ from \eqref{mu}.

\begin{rmk}
In many applications, hyperparameters of the covariance function are fit 
to the data using methods that require the gradient of the Gaussian 
process likelihood function in addition to the determinant of the matrix
$K + \sigma^2I$. 
Methods for the evaluation of these quantities is the subject of a large 
literature including, for example, \cite{minden1, gardner1, wenger2021}. 
We leave the evaluation of these quantities using 
equispaced Fourier representations to future work.
\end{rmk}

\section{Error analysis}\label{s6}

In this section we bound the kernel approximation error
due to discretization,
state choices of numerical parameters that achieve this bound for the SE and Mat\'ern
kernels,
then bound the error (bias) in the posterior mean due to both kernel and linear
system solution errors,
relative to ``exact'' inference using the true kernel.

\subsection{Kernel discretization error}
\label{s:discr}

Aside from its computational advantages,
the equispaced choice of Fourier frequencies has precise yet elementary
error estimates. 
The following applies to any integrable
and somewhat smooth kernel, in exact arithmetic.
\begin{pro}[Pointwise kernel error]  
  \label{p:pw_err}
  Suppose that the covariance kernel $k: \R^d \to \R$
  and its Fourier transform $\hat{k}$
  decay uniformly as
  $|k(x)|\le C(1+\|x\|)^{-d-\delta}$ and
  $|\hat{k}(\xi)|\le C(1+\|\xi\|)^{-d-\delta}$ for some $C$, $\delta>0$.
  Let $h>0$, $m\in\N$, then define $\tilde k$ by the discrete sum \eqref{110}.
  Then for any $x \in \R^d$ we have
  \begin{align}
    \label{ptwise_err}
  \tilde{k}(x) - k(x)
  \;\;=\;\;
  -
  \underbrace{\sum_{n \in \Z^d, \, n \neq \mbf{0}} k\left(x + \frac{n}{h} \right)}
  _\text{aliasing error}
  \;\;+\;\;
  \underbrace{h^{d} \sum_{j \in \mathbb{Z}^{d}, \, j \not \in J_{m}} \hat{k}(j h)e^{2\pi i h \inner{j}{x}}}_\text{truncation error}
  \,.
\end{align}
\end{pro}
\begin{proof}
  Writing \eqref{110} using the argument $x$ instead of $x-x'$ gives
  $\tilde{k}(x) = h^{d} \sum_{j \in J_{m}} \hat{k}(j h)e^{2\pi i h  \inner{j}{x} }$.
  Subtracting from this expression
  a shifted and scaled version of the Poisson summation formula
  \cite[Ch.~VII, Cor.~2.6]{steinweissbook},
\begin{align}\label{psf}
h^{d} \sum_{j \in \mathbb{Z}^{d}} \hat{k}(jh)e^{2\pi h i \inner{j}{x}} = 
\sum_{n \in \Z^{d}} k\biggl(x + \frac{n}{h} \biggr),
\end{align}
and separating the $n=\mbf{0}$ term, completes the proof.
\end{proof}

This proposition states that the kernel approximation error has two contributions:
\ben
\item
The first term is interpreted as the {\bf aliasing error} due to
the nonzero quadrature node spacing $h$, and can be made
small by sending $h$ to zero.
It takes the form of a punctured sum over periodic images of the kernel $k$,
thus is sensitive to $x$.
To make use of it, a uniform estimate over $x\in[-1,1]^d$ is necessary,
this being the set that $x-x'$ takes when $x$ and $x'$ range over the physical domain
$D = [0,1]^d$.
For a kernel with rapid decay over a scale $\ell \ll 1$, this aliasing error
is dominated by the images nearest to $[-1,1]^d$, so is uniformly
small once $h^{-1}$ exceeds $1$ by a few times $\ell$.
\item
The second term on the right-hand side of \eqref{ptwise_err} is the
{\bf truncation error} due to restricting the Fourier integral to the finite box
$[-mh,mh]^d$.
It can be made uniformly small by choosing $mh$ large enough so that
the tail integral of $|\hk|$ is small over the exterior of this box,
and discarding the phases $e^{2\pi i h\inner{j}{x}}$.
\een
In practice we first set $h$ to achieve a certain uniform aliasing error, then
choose $m$ according to the decay of $\hk$ to achieve a uniform
truncation error of similar size.

\subsection{Choice of EFGP parameters $h$ and $m$ for two common kernels}
\label{s:choice}

Both the squared-exponential kernel (see \eqref{Gker}) and Mat\'ern
kernel (with $\nu\ge 1/2$, see \eqref{Cker})
satisfy the decay conditions in Proposition~\ref{p:pw_err}.
Thus uniform bounds on $|\tilde{k}(x)-k(x)|$
over $x\in[-1,1]^d$ are possible by bounding the
aliasing and truncation errors as outlined above.
The analysis is detailed, using integral bounds on lattice
sums, bounds on modified Bessel functions, and induction in dimension $d$,
and demands a separate work \cite{efgp_anal}.
Here we summarize the resulting
choices of $h$ (Fourier grid spacing) and $m$ (grid half-size per dimension)
to guarantee any uniform kernel error tolerance $\eps>0$.
In the Mat\'ern case we then state a heuristic choice that leads to better
performance.

Owing to the rapid decay of the SE kernel and its Fourier transform,
the aliasing error and the
truncation error can be bounded by the term with the smallest $|n|$ and smallest $|m|$ in \cref{ptwise_err}
respectively, up to a $d$-dependent constant.
In particular, we set $h=\big(1 + \ell \sqrt{2 \log (4d\,3^d/\eps)}\big)^{-1}$, then $m=\lceil\sqrt{\half \log (4^{d+1}d/\eps)}/\pi\ell h \rceil$, which 
guarantees that aliasing and truncation errors are each uniformly bounded by $\eps/2$.

On the other hand, the slow algebraic decay $\hat{k}(\xi) =\bigO( \|\xi\|^{-2\nu-d})$
of the Fourier transform of the Mat\'ern kernel
demands large $m$, particularly for small $\nu$.
In order to guarantee uniform error $\eps$, one would require
\be
h\leq \big(1 + \ell \sqrt{2d/\nu} \log (d\,3^d/\eps) \big)^{-1} , \quad \textrm{and} \quad m\geq (d\,5^{d-1}/\pi^{d/2} \eps)^{1/2\nu} \cdot (1.6)\sqrt{\nu}/\pi h \ell.
\label{maternhm}
\ee
However, the bound for the accuracy of the posterior mean in terms of the uniform
error of the covariance kernel is not sharp and tends to be an overestimate,
leading to an unnecessarily large $m$.
This motivates a heuristic $L^2$ kernel approximation error,
which, combined with extensive numerical 
experiments, leads us when $1/2 \leq \nu \leq 5/2$ to instead set
\be
h=\bigl(1+0.85 (\ell/\sqrt{\nu}) \log 1/\tilde{\eps} \bigr)^{-1} , \quad \textrm{and}  \quad m=\left\lceil \frac{1}{h} \biggl( \pi^{\nu+d/2} \ell^{2\nu} \frac{\tilde{\eps}}{0.15}\biggr)^{-1/(2\nu + d/2)} \right\rceil,
\label{maternheur}
\ee
where $\tilde{\eps} := \eps \|k\|_{L^2[-1,1]^d}$
rescales the error so that $\eps$ is interpreted as
the {\em relative} precision of $L^2$ kernel approximation.
For fixed $\nu$, $\ell$, and $h$, this estimate
implies $m=\bigO(1/\eps^{1/(2\nu + d/2)})$, as compared to the uniform estimate
\eqref{maternhm}
for which $m=\bigO(1/\eps^{1/2\nu})$; the former is more forgiving (hence computationally efficient) for small values of $\nu$, particularly in two or three dimensions.
For $\nu>5/2$ we revert to \eqref{maternhm}.

The heuristic choice \eqref{maternheur}
relies on Monte Carlo approximation of the Frobenius
norm $\|\tilde{K} - K\|_F$ in terms of $\|k\|_{L^2[-1,1]^d}$,
requiring assumptions on the distribution of points $\{x_n\}$.
See \cite{efgp_anal} for a derivation and numerical tests of the estimate.
In practice, the above choices of $h$ and $m$ result in closer agreement of
EEPM with the user-requested requested tolerance $\eps$.

\subsection{Error analysis of posterior mean prediction}
\label{s:meanacc}

In this section, we bound the error in the posterior mean $\mu$ 
(relative to exact inference) induced by two sources of error: kernel approximation,
and the residual when solving the weight space linear system.
This adds to the literature on GP errors (bias) induced by an incorrect
GP prior covariance, which is generally asymptotic and specific
to some random model for data points
(e.g., \cite{stein1990,  sanz-alonso2021, vandervaart2008}).

We start by estimating the spectral norm $\|\tilde K - K\|$ of the difference
between the true covariance kernel matrix $K$ and its finite-rank
equispaced Fourier approximation $\tilde K = \X\Xt$,
for arbitrary data points. 
\begin{pro}   
  \label{p:Kerr}
  Let the data points be $x_1,\dots,x_N \in D$.
  Let $K$ be their covariance matrix using the true kernel $k(x,x')$,
  and $\tilde K = \X \Xt$ be its rank-$M$
  approximation using a kernel $\tilde k(x,x') = \sum_{j=1}^M \phi_j(x)\overline{\phi_j(x')}$,
  with uniform approximation error $|k(x,x')-\tilde k(x,x')| \le \eps$ for all $x,x'\in D$.
  Then
  \be
  \|\tilde K - K\| \le \|\tilde K-K\|_F \le N \eps~.
  \label{Kerr}
  \ee
\end{pro}
The proof simply sums the $N^2$ matrix elements of size at most $\eps$,
then bounds the spectral norm
by the Frobenius norm $\|\cdot\|_F$ \cite[\S~2.3.2]{golubvanloan}.
For equispaced Fourier approximation,
Proposition~\ref{p:Kerr} immediately applies
using the $\eps$ (sum of aliasing and truncation errors) from Sections~\ref{s:discr}--\ref{s:choice}.

Although EFGP computes in weight space, the following will be useful,
showing that the above spectral error $\|\tilde K - K\|$ controls the
relative error in the function space coefficients.

\begin{lem}[Stability under covariance matrix perturbation]\label{l:cov_perturb}
  Let $\mbf{y} \in \R^N$, let $K$, $\tilde{K}$ be
  positive semi-definite matrices, let $\sigma>0$,
  and let $\bal,\tbal \in \R^N$ be the unique solutions to 
  \be
  \label{370}
(K + \sigma^2 I) \bal = \mbf{y} \text{,} \qquad (\tilde{K} + \sigma^2 I) \tbal = \mbf{y}~.
\ee
Then the relative norm of the difference in solution vectors is bounded by
\be
  \| \tbal - \bal \| / \| \bal \|
  \;\le\;
  \| (\tilde K + \sigma^2 I)^{-1} \| \, \| K - \tilde{K} \|
  \;\le\;
  \sigma^{-2} \| K - \tilde{K} \|
  ~.
  \label{alerr}
\ee
where $\| \cdot \|$ denotes the $\ell^2$-norm for vectors and the spectral norm for 
matrices.
\end{lem}
\begin{proof}
  Eliminating $\mbf{y}$ in \eqref{370} gives
  $
  (\tilde K + \sigma^2 I)
  (\bal-\tbal) = (\tilde K - K)\bal
  $.
  Left-multiplying by $(\tilde K + \sigma^2 I)^{-1}$ then using the properties
  of norms gives the first inequality. The second inequality follows
  by diagonalizing $\tilde K$, whose eigenvalues must be nonnegative.
\end{proof}

Finally, our main result bounds the posterior mean error induced by
approximating the kernel with uniform error $\eps$ and by
solving the weight-space linear system to relative residual $\rho$.
Recall that the posterior mean may be evaluated either
i) at the given data points $\{x_n\}_{n=1}^N$,
or ii) at new data points (denoted by $x$ in the Introduction).
In practice (see Algorithm~\ref{a1})
both tasks are performed by evaluating \eqref{muws},
given $\hat{\bbe}$, a weight space vector approximately solving \eqref{be}.
Since improved error estimates are possible for task i), we give separate
statements.
It is also now convenient to assume that the kernel has been
scaled to (at most) unit prior variance.

\begin{theorem}[Stability of posterior means] 
  \label{t:muerr}
  Let $\mbf{y}\in\R^N$, $\sigma>0$, let $k$ be a positive covariance kernel
  with $|k(x)|\le 1$, $\forall x\in\R^d$,
  with $K$ and its rank-$M$ approximation $\tilde K=\X\Xt$
  as in Proposition~\ref{p:Kerr} with uniform error $\eps\ge0$,
  and let $\rho\ge0$ be the weight space relative residual norm, that is,
  the numerical coefficient vector $\hat{\bbe}\in\C^M$
  satisfies
  \be
  (\tilde K + \sigma^2I)\hat{\bbe} \;=\; \Xt\mbf{y} + \mbf{r},
  \qquad \mbox{ for some residual $\mbf{r}$ with }
  \quad \|\mbf{r}\|/\|\Xt \mbf{y}\| = \rho.
  \label{beresid}
  \ee
  i) Let $\bmu$ be the vector of true posterior means at the data points,
  and $\hbmu = \X \hat{\bbe}$ be its EFGP approximation as in \eqref{bmu}.
  Then its error (relative to the data magnitude) obeys
\be\label{374}
  \frac{\| \hat\bmu-\bmu \|}{\| \mbf{y} \|}
  \;\le\;
  \frac{N}{\sigma^2} (\eps + \rho + \eps\rho)
  \qquad \mbox{ (relative error of posterior mean at the data points).}
\ee
ii) Let $\mu(\xz)$ be the true posterior mean at a new target $\xz\in\R^d$,
and let $\hat\mu(\xz) = \sum_{j=1}^M \hat{\beta}_j\phi_j(x)$
be its EFGP approximation.
Then its error relative to the data magnitude obeys
\be\label{muerrnew}
\frac{|\hat\mu(\xz)-\mu(\xz)|}{\|\mbf{y}\|/\sqrt{N}}
\;\le\;
\frac{N^2}{\sigma^4}\eps + \frac{N}{\sigma^2} ( \eps + \rho + \eps\rho)
\qquad \mbox{(relative error at a test target).}
\ee
\end{theorem}
\begin{proof}
  We first bound the effect of the residual: let
  $\tilde\bbe$ be the unique weight space solution
  to $(\tilde K+\sigma^2I)\tilde\bbe = \Xt\mbf{y}$, then by subtraction of this
  equation from \eqref{beresid} we bound
  \be
  \|\hat{\bbe}-\tilde\bbe\| \;\le\;
  \|(\tilde K+\sigma^2I)^{-1}\|\, \|\mbf{r}\|
  \;\le\;
  \sigma^{-2} \cdot \rho\|\Xt\mbf{y}\|.
  \label{beer}
  \ee
  To handle task i), we apply the above to get the bound
  \be
  \|\X(\hat{\bbe}-\tilde\bbe)\| \;\le\;
  \|\X\| \cdot \sigma^{-2} \rho\|\Xt\mbf{y}\|
  \;\le\; 
  \frac{\rho}{\sigma^2} \|\X\|^2 \|\mbf{y}\|
  \;\le\;
  \frac{N}{\sigma^2} (1+\eps)\rho \|\mbf{y}\|,
  \label{Xbeer}
  \ee
  where we used $\|\Xt\|^2 = \|\X\|^2 = \|\tilde K\|$, which can be checked
  by inserting the singular-value
  decomposition $\X = U\Sigma V^*$ into $\tilde K = \X\Xt$,
  then used $\|\tilde K\|\le N(1+\eps)$.
  The latter follows since $\tilde K = (\tilde K-K) + K$
  so that $\|\tilde K\| \le \|\tilde K - K\| + \|K\|$,
  then one applies Proposition~\ref{p:Kerr} and the bound
  $\|K\| \le \|K\|_F \le N$ coming from the unit bound on matrix entries.

  Now letting $\bal$ and $\tbal$ be the unique solutions to the function space systems
  in \eqref{370}, then $\bmu=K\bal$, so the mean error vector at the data points is
  \be
  \hat{\bmu} - \bmu \;=\;
  \X \hat{\bbe} - K\bal
  \;=\;
  \X (\hat{\bbe} -\tilde{\bbe}) + (\X\tilde\bbe - K\bal).
  \label{bmuerr}
  \ee
  By weight-space equivalence
  (Lemma~\ref{l:equiv} applied to $\tilde X$) we have $\X\tilde\bbe = \tilde K \tbal$,
  so that the second bracketed term becomes
  $\tilde K \tbal - K\bal$, which equals $\sigma^2(\bal-\tbal)$ using \eqref{370}.
  We bound its norm by combining Lemma~\ref{l:cov_perturb} then Proposition~\ref{p:Kerr}
  to get
  \be
  \sigma^2 \|\bal-\tbal\| \;\le\;
  \sigma^2 \cdot \sigma^{-2} \|\tilde K - K\|\, \|\bal\|
  \;=\; N\eps \|\bal\|
  \;\le\;
  \frac{N}{\sigma^2} \eps \|\mbf{y}\|,
  \label{balerr}
  \ee
  where in the last step we used a naive coefficient norm bound
  $\|\bal\| \le \| (K + \sigma^2 I)^{-1} \| \|\mbf{y}\| \le \sigma^{-2} \|\mbf{y}\|$.
  Applying the triangle inequality to \eqref{bmuerr},
  the first term is bounded by \eqref{Xbeer} and the second bracketed term by \eqref{balerr}.
  Adding the two terms proves \eqref{374}.

  Turning to task ii), we recall from \eqref{mu} that $\mu(\xz) = \kz\tran\bal$
  where $\kz := [k(\xz, x_1), \dots, k(\xz, x_N)]\tran$. 
  We also define the vector $\bpz := [\phi_1(\xz), \dots,\phi_M(\xz)]\tran$.
  We start with the identity
  $$
  \hat\mu(\xz)-\mu(\xz)\;=\;
  \bpz\tran\hat\bbe -  \kz\tran\bal
  \;=\;
  \bpz\tran(\hat\bbe - \tilde\bbe)
  + (\bpz\tran\tilde\bbe - \kz\tran\tbal)
  + \kz\tran(\tbal - \bal).
  $$
  By weight-space equivalence, $\bpz\tran\tilde\bbe = \tkz\tran\tbal$.
  Applying the triangle then Cauchy--Schwarz inequalities to the three
  terms thus gives
  $$
  |\hat\mu(\xz)-\mu(\xz)|  \;\le \;
  \|\bpz\| \|\hat{\bbe}-\tilde\bbe\| +
  \|\tkz - \kz\| \|\tbal\| +
  \|\kz\| \|\tbal-\bal\|.
  $$
  Tackling the first term, we combine \eqref{beer} with
  $\|\bpz\|^2 = \sum_{\bj\in J_m} h^d \hk(h\bj) = \hat{k}(\mbf{0}) \le
  k(\mbf{0}) + \eps \le 1+\eps$, and the above bound on $\|\Xt\|$
  to get its bound $\sigma^{-2}\sqrt{N}\|\mbf{y}\|(1+\eps)\rho$.
  For the second term we combine
  $\|\tkz - \kz\| \le \sqrt{N}\eps$,
  which follows from the uniform kernel approximation,
  with $\|\tbal\| \le \sigma^{-2}\|\mbf{y}\|$,
  to get its bound $\sigma^{-2}\sqrt{N}\|\mbf{y}\| \eps$.
  The last term uses $\|\kz\|\le\sqrt{N}$ and \eqref{balerr},
  and ends up dominating the other two terms by a factor $N/\sigma^2$.
  Summing the three terms and 
  dividing by $\|\mbf{y}\|/\sqrt{N}$ proves \eqref{muerrnew}.
\end{proof}
 
We remind the reader that, although the above proof made use of
the function space vectors $\bal$ and $\tbal$, the numerical
method works entirely in weight space, for which linear system
$\hat\bbe$ is an approximate solution vector.
The appearance of $\rho$ with a similar factor as $\eps$ in
Theorem~\ref{t:muerr} justifies our
choice of $\rho \approx \eps$ in Algorithm~\ref{a1}.

Note that the
$\eps$-dependent error bound
for new targets (task ii) is a factor $N/\sigma^2$ worse than for
prediction at the given data points (task i).
The underlying reason is that for task i) the ``norm-reducing'' trick
$\tilde K \tbal- K\bal = \sigma^2(\bal-\tbal)$, due to the difference
of linear systems \eqref{370}, has no known analog in task ii) for terms like
$\tkz\tran\tbal - \kz\tran\bal$ or even
$\tkz\tran(\tbal - \bal)$.
By adding conditions (e.g., iid random) on the data points, it is possible
that an approximate orthogonality
$\|\Xt(\tilde K-K)\| \ll \|\X\|\|\tilde K -K\|$,
and/or a smaller $\ell$-dependent bound on $\|\kz\|$, could be exploited.
Yet, in Section~\ref{s:numerics} we will see that
in practice Fourier approximation of relatively
smooth kernels causes {\em no worse error at new targets} than at the data points
themselves.

\section{Numerical experiments on GP regression}
\label{s:numerics}

In this section we benchmark EFGP, the proposal of this paper, in four
ways.
Section~\ref{s:conv} tests the accuracy and scaling of EFGP itself;
Section~\ref{s:comp} compares the run-time vs accuracy
tradeoffs for EFGP against those of three state-of-the-art other methods;
Section~\ref{s:co2} demonstrates EFGP in a real 2D data application; finally
Section~\ref{s:big} tests the scaling of EFGP for massive data sizes
$N$ (up to a billion) while comparing its performance to the performance of its nearest competitor.
Our numerical experiments use one of the following two common
families of kernels \cite{rasmus1}:
\bi
\item The squared exponential kernel with length scale $\ell$
  (using $|\cdot|$ for Euclidean norm),
\begin{equation}
  G_\ell(x) :=  \exp{\left(-\frac{|x|^2}{2\ell^2} \right)}.
  \label{Gker}
\end{equation}
\item
The Mat\'ern kernel with smoothness parameter $\nu\ge 1/2$ and length scale $\ell$,
\begin{equation}
  C_{\nu,\ell}(x) := \frac{2^{1-\nu}}{\Gamma(\nu)} \left(\sqrt{2\nu} \frac{|x|}{\ell} \right)^{\nu} K_{\nu} \biggl(  \sqrt{2\nu}\frac{|x|}{\ell}\biggl),
  \label{Cker}
\end{equation}
where $K_{\nu}$ is the modified Bessel function of the second kind.
\ei

In all sections apart from Section~\ref{s:co2}, we use synthetic data generated in the following manner.
We use as an underlying test function the simple plane wave
\begin{equation}
\label{eq:fdef}
f(x) = \cos (2\pi \inner{x}{\omega}+ 1.3),  \qquad x \in [0,1]^d,
\end{equation}
where $\omega\in\R^d$ is some wavevector
with length $|\omega| \in [3, 6.2]$ depending on the experiment.
Thus $f$ has several oscillations across the domain,
while also remaining smooth enough for
consistency with the prior covariance model for the tested kernel
length scale $\ell = 0.1$.
The constant $1.3$ imparts a generic phase shift.
We have tried other smooth $f$ functions and find little difference
in the conclusions as long as $f$ is not too oscillatory.

Training points $x_1,\dots,x_N \in [0,1]^d$ are chosen uniformly iid, from
which noisy data $y_{n} = f(x_{n}) + \epsilon_{n}$ is generated
with iid draws $\epsilon_{n} \sim \mathcal{N}(0,\sigma^2)$.
Thus the noise level is consistent with $\sigma$ used for GP regression.
The variance of $y_n$ data values (for $\sigma<1$) is around 1.

Test targets $\xt_1,\dots,\xt_{N_t} \in [0,1]^d$
(here $\ast$ denotes ``test'' rather than conjugate),
are chosen to lie on the $d$-dimensional tensor product of the 1D uniform grid
$\{0,1/n_t,2/n_t,\dots,(n_t-1)/n_t\}$, so that $N_t = n_t^d$.
We will specify $n_t$ for each experiment below.
Synthetic held-out noisy test data is then generated via
$\yt_{n} = f(\xt_{n}) + \epsilon_{n}$ with $\epsilon_n$ drawn as above.

We will report two types of error metrics:
\ben
\item EEPM (estimated error in posterior mean).
  This is the difference (or bias) between a computed
  posterior mean $\hat\mu$
  and an estimate of the reference posterior mean $\mu$ from exact GP inference.
  For small $N\le 10^4$ we used a dense direct solve in function space
  (the traditional definition of ``exact'') to get the reference $\mu$,
  while for larger $N$ we used EFGP with tolerance $\eps$ pushed small enough to
  insure convergence to a negligible error (giving a ``self-convergence'' accuracy estimate, standard in numerical analysis).
  We report root mean square error, and use the subscript ``new'' to indicate the new held-out test targets,
  otherwise the training points are assumed:
  \[
  \textrm{EEPM} := \biggl(\frac{1}{N} \sum_{n=1}^{N} [\hat\mu(x_{n}) - \mu(x_{n})]^2 \biggl)^{1/2} \!\!\!\! ,
  \quad
  \textrm{EEPM}_\text{new} := \biggl(\frac{1}{N_t} \sum_{n=1}^{N_t} [\hat\mu(\xt_{n}) - \mu(\xt_{n})]^2 \biggl)^{1/2}  \!\! .
  \]
  These two measures of accuracy are precisely, in the case of data $y_n$ of unit variance, the quantities bounded by Theorem~\ref{t:muerr}.

\item RMSE (root mean square data prediction error).
  This is a standard measure of the predictive ability of a particular GP regression computation.
  It is only meaningful for the held-out test points $\xt_n$.
  With no subscript it indicates a tested method (using
  $\hat\mu(x)$ as the predictor), whereas
  with the subscript ``ex'' it indicates exact GP inference
  (using the reference $\mu(x)$ as the predictor):
  \[
  \textrm{RMSE} := \biggl(\frac{1}{N_t} \sum_{n=1}^{N_t} [\hat\mu(\xt_{n}) - \yt_n]^2 \biggl)^{1/2} \!\!\!\! ,
  \qquad
  \textrm{RMSE}_\text{ex} := \biggl(\frac{1}{N_t} \sum_{n=1}^{N_t} [\mu(\xt_{n}) - \yt_n]^2 \biggl)^{1/2}  \!\! .
  \]
  Here the reference $\mu(x)$ is computed as for EEPM.
\een

It is worth motivating inclusion of the EEPM metric.
While RMSE is ubiquitous in the literature,
it is not a very precise indicator of whether an approximate method
is solving the GP problem
correctly,
because it does not converge to zero as the method approaches exact inference.
Rather, for a well-specified model,
it is expected to tend to $\sigma$ as $N\to\infty$ and $N_t\to\infty$,
as with exact GP inference.
Thus, while it {\em can} be used as a high-accuracy convergence test
\cite[Fig.~8]{oneil1}, it tells only part of the story.
For example, in GP applications where the kernel locally averages over many data
points, the expected variation in the posterior mean can be $\ll \sigma$,
but RMSE (on noisy held-out data) does not measure this sub-$\sigma$ accuracy.
In contrast, EEPM explicitly measures the approximation accuracy for a given GP
regression implementation.

All numerical experiments were run on a 2.6 GHz Intel i7-8850H MacBook Pro
(6 physical cores) with $16$ GB RAM, running MATLAB R2021b,
with the exception of the results in Section~\ref{s:big} which were run on a
desktop with two 3.4 GHz Intel Xeon Gold 6128 Skylake CPUs
(12 physical cores total) and $192$ GB RAM, running R2021a.

Our MATLAB/Octave
implementation of EFGP, wrappers to the three other methods tested, and all
scripts for generating the results in this
paper are available at~\url{https://github.com/flatironinstitute/gp-shootout}.

\begin{table}[ht!]
\centering
  \resizebox{\textwidth}{!}{%
\begin{tabular}{ccccccccccccc}
   $N$ & $d$ & kernel & $\varepsilon$ & $m$ & pre (s) & solve (s) & mean (s) & tot (s)  & iters & $\textrm{EEPM}_{\textrm{new}}$ & RMSE & RMSE$_{\textrm{ex}}$\\ 
\hline
$10^{3}$ &$1$ &SE &$10^{-4}$ &$13$ & $0.002$ & $ 0.001 $ & $ 0.001 $ & $ 0.003 $ & $15$ & $4.9\times 10^{-4}$ & $3.0\times 10^{-1}$ & $3.0\times 10^{-1}$ \\ 
$10^{5}$ &$1$ &SE &$10^{-4}$ &$13$ & $0.010$ & $ 0.002 $ & $ 0.004 $ & $ 0.016 $ & $15$ & $3.2\times 10^{-3}$ & $3.0\times 10^{-1}$ & $3.0\times 10^{-1}$ \\ 
$10^{7}$ &$1$ &SE &$10^{-4}$ &$13$ & $0.802$ & $ 0.002 $ & $ 0.333 $ & $ 1.137 $ & $19$ & $1.3\times 10^{-3}$ & $3.0\times 10^{-1}$ & $3.0\times 10^{-1}$ \\ 
\hline \hline 
$10^{3}$ &$2$ &SE &$10^{-4}$ &$16$ & $0.007$ & $ 0.013 $ & $ 0.003 $ & $ 0.022 $ & $64$ & $1.2\times 10^{-4}$ & $3.3\times 10^{-1}$ & $3.3\times 10^{-1}$ \\ 
$10^{5}$ &$2$ &SE &$10^{-4}$ &$16$ & $0.020$ & $ 0.035 $ & $ 0.009 $ & $ 0.063 $ & $184$ & $9.2\times 10^{-4}$ & $3.1\times 10^{-1}$ & $3.1\times 10^{-1}$ \\ 
$10^{7}$ &$2$ &SE &$10^{-4}$ &$16$ & $1.478$ & $ 0.029 $ & $ 0.711 $ & $ 2.217 $ & $138$ & $1.8\times 10^{-3}$ & $3.1\times 10^{-1}$ & $3.1\times 10^{-1}$ \\ 
\hline \hline 
$10^{3}$ &$3$ &SE &$10^{-3}$ &$17$ & $0.074$ & $ 0.207 $ & $ 0.022 $ & $ 0.304 $ & $34$ & $1.3\times 10^{-3}$ & $4.1\times 10^{-1}$ & $4.1\times 10^{-1}$ \\ 
$10^{5}$ &$3$ &SE &$10^{-3}$ &$17$ & $0.107$ & $ 0.919 $ & $ 0.035 $ & $ 1.060 $ & $130$ & $3.4\times 10^{-3}$ & $3.0\times 10^{-1}$ & $3.0\times 10^{-1}$ \\ 
$10^{7}$ &$3$ &SE &$10^{-3}$ &$17$ & $2.583$ & $ 0.698 $ & $ 1.315 $ & $ 4.596 $ & $80$ & $6.9\times 10^{-3}$ & $3.0\times 10^{-1}$ & $3.0\times 10^{-1}$ \\ 
\hline \hline 
$10^{3}$ &$1$ &Mat $1/2$ &$10^{-4}$ &$1555$ & $0.016$ & $ 0.059 $ & $ 0.002 $ & $ 0.076 $ & $53$ & $2.0\times 10^{-3}$ & $3.1\times 10^{-1}$ & $3.1\times 10^{-1}$ \\ 
$10^{5}$ &$1$ &Mat $1/2$ &$10^{-4}$ &$1555$ & $0.016$ & $ 0.156 $ & $ 0.006 $ & $ 0.178 $ & $155$ & $1.7\times 10^{-2}$ & $3.0\times 10^{-1}$ & $3.0\times 10^{-1}$ \\ 
$10^{7}$ &$1$ &Mat $1/2$ &$10^{-4}$ &$1555$ & $1.097$ & $ 0.065 $ & $ 0.415 $ & $ 1.577 $ & $69$ & $6.8\times 10^{-3}$ & $3.0\times 10^{-1}$ & $3.0\times 10^{-1}$ \\ 
\hline \hline 
$10^{3}$ &$2$ &Mat $1/2$ &$10^{-3}$ &$97$ & $0.064$ & $ 0.285 $ & $ 0.004 $ & $ 0.353 $ & $41$ & $1.4\times 10^{-2}$ & $3.5\times 10^{-1}$ & $3.5\times 10^{-1}$ \\ 
$10^{5}$ &$2$ &Mat $1/2$ &$10^{-3}$ &$97$ & $0.058$ & $ 1.147 $ & $ 0.013 $ & $ 1.217 $ & $147$ & $5.8\times 10^{-2}$ & $3.1\times 10^{-1}$ & $3.2\times 10^{-1}$ \\ 
$10^{7}$ &$2$ &Mat $1/2$ &$10^{-3}$ &$97$ & $2.306$ & $ 0.494 $ & $ 1.223 $ & $ 4.023 $ & $53$ & $2.6\times 10^{-2}$ & $3.1\times 10^{-1}$ & $3.1\times 10^{-1}$ \\ 
\hline \hline 
$10^{3}$ &$3$ &Mat $1/2$ &$5\times 10^{-3}$ &$23$ & $0.349$ & $ 1.307 $ & $ 0.082 $ & $ 1.737 $ & $22$ & $5.6\times 10^{-2}$ & $4.6\times 10^{-1}$ & $4.8\times 10^{-1}$ \\ 
$10^{5}$ &$3$ &Mat $1/2$ &$5\times 10^{-3}$ &$23$ & $0.311$ & $ 2.815 $ & $ 0.076 $ & $ 3.201 $ & $88$ & $1.0\times 10^{-1}$ & $3.2\times 10^{-1}$ & $3.3\times 10^{-1}$ \\ 
$10^{7}$ &$3$ &Mat $1/2$ &$5\times 10^{-3}$ &$23$ & $4.090$ & $ 1.315 $ & $ 1.655 $ & $ 7.059 $ & $38$ & $4.6\times 10^{-2}$ & $3.0\times 10^{-1}$ & $3.0\times 10^{-1}$ \\ 
\hline \hline 
\end{tabular}
}%
  \caption{Run-time and accuracy for GP regression with EFGP on simulated data, for two kernels with $\ell=0.1$ and $\sigma=0.3$ (see Section~\ref{s:conv}).
  $N$ is the training data size, $d$ the dimension, $\eps$ the tolerance, and $m$ the number of positive Fourier nodes per dimension.
    ``pre'' is the CPU time for precomputation (steps 2--3 of Algorithm~\ref{a1}), ``solve'' is for step 4, ``mean'' is for step 6 evaluating at a grid of $60^d$ new targets, and ``tot'' the total.
    }
\label{t:efgp_only}
\end{table}

\subsection{Accuracy and CPU-time performance of EFGP\label{s:conv}}

Here we test EFGP for computing the posterior mean in GP regression.
We test two covariance kernels $k$:
the squared exponential (SE) defined in \eqref{Gker},
being an extremely smooth kernel
(rapid decay in Fourier space), and the
Mat\'ern kernel of \eqref{Cker}
with $\nu = 1/2$, being the least smooth of that family (slow algebraic Fourier
decay).
For Fourier methods the latter is thus the most challenging within the Mat\'ern family.
For all kernels the length scale $\ell=0.1$ was used,
and the prior variance was $k(\mbf{0})=1$.

We implemented EFGP, as described in Algorithms~\ref{a:FF} and \ref{a1},
in MATLAB/Octave, calling v2.1.0
of the FINUFFT package \cite{barnett1} for all nonuniform FFTs
using a NUFFT tolerance of $\eps/10$, where $\eps$ is the
EFGP tolerance.
Parameters $h$ and $m$ were chosen as in Section~\ref{s:choice}.
Both MATLAB (in particular for FFTs) and FINUFFT exploit multithreading.

We generated synthetic data as described above,
with wavevectors $\omega=3$ in 1D, $\omega=(3,6)/\sqrt{5}$ in 2D,
or $\omega = (3,9,6)/\sqrt{14}$ in 3D.   
We set $\sigma = 0.3$ for both additive noise and GP regression.
The test target grid has $n_{t} = 60$ nodes per dimension,
i.e.\ $N_t = 60^d$ total test points.
The tolerance $\eps$ was chosen as $10^{-4}$ by default, but
was increased for $d=3$, and also for $d=2$ with the Mat\'ern kernel,
to give more reasonable run-times for the reference computations.
The conjugate gradient relative residual stopping criterion $\rho$ was
set equal to the tolerance $\eps$ (we checked that lowering this caused
little improvement).

The results are in Table~\ref{t:efgp_only}, whose columns are defined in the caption.
Each of the six groups has a particular dimension, kernel, and tolerance
(and thus $m$). Within each group $N$ grows.
Notice that all total CPU times, even for $N=10^7$
(in prior literature considered unsolvably big)
are {\em only a couple of seconds on a laptop}.
They are dominated by the solve time (CG iteration using FFTs),
thus scale with $M$, hence are longer for higher $d$ and the non-smooth kernel (Mat\'ern).
Since $M$ remains constant when $N$ varies,
the solve times {\em per iteration} are also independent of $N$.
Precomputation times do scale with $N$,
asymptoting to linear with a throughput of $3\times 10^6$ to $10^7$ points per second.

For the SE kernel, the posterior mean accuracy EEPM$_{\textrm{new}}$
is almost always within one digit (a factor of ten) of the requested $\eps$.
For the more challenging Mat\'ern,
the EEPM$_{\textrm{new}}$ was larger than $\eps$ by one to two digits.
(The upper bound of $N^2 \eps / \sigma^4$ of Theorem~\ref{t:muerr}
is unfortunately useless here since it varies from $10^5$ to $10^{14}$.)

The RMSE is very close to the noise $\sigma=0.3$ (as it should be
when GP regression is a success),
except at smaller data sizes ($N\le 10^4$) when it can grow up to $0.48$.
Crucially, for all experiments, there is {\em no significant
difference} between RMSE and RMSE$_{\textrm{ex}}$,
thus EFGP at the chosen $\eps$ is essentially as good a predictor as the reference
(``exact'') GP.

Finally, the iteration counts seem first to grow with $N$, then to decrease
again at the highest $N$. 
The growth in the number of iterations as a function of $N$ is due to the 
increase in the condition number of the weight-space matrix. We observe
that condition numbers of the weight-space system are roughly $N/\sigma^2$
independent of the kernel, for small enough $\eps$. 
We leave a rigorous study of condition number to \cite{efgp_anal}. 
The decrease could be accounted for by growth in the lowest eigenvalue
of $\Xt\X$ relative to that of the true kernel $K$.

\begin{figure*}[ht!]
\centering
\begin{subfigure}{0.480\linewidth}
  \centering
\begin{tikzpicture}[scale=0.7]
\centering
\begin{axis}[
    ymin= -3.2, ymax=2,
    xmin=-13, xmax=0,
    ytick={-2, -1, 0, 1, 2},
    xlabel=$\log_{10} \textrm{EEPM}_{\textrm{new}}$, 
    xtick={-12, -10, -8, -6, -4, -2, 0}, 
    ylabel=$\log_{10} $ time (s),
    legend pos = south west,
    x dir=reverse
  ]
  
\addplot[color = red, mark=square*, line width=0.1mm]
    coordinates
    {
(-2.04, -2.13)
(-3.12, -2.17)
(-3.66, -2.23)
(-5.58, -2.22)
(-5.64, -2.1)
(-6.8, -1.98)
(-8.94, -1.95)
(-10.7, -1.99)
(-11.4, -1.92)
(-11.5, -1.85)
    };
    
\addplot[color = blue, mark=square*, line width=0.1mm]
    coordinates 
    {
(-0.331, -0.0396)
(-4.17, 0.18)
(-6.36, 0.201)
(-5.55, 0.311)
    };
  
\addplot[mark=square*, line width=0.1mm]
    coordinates 
    {
(-2.53, 0.382)
(-3.88, 0.382)
(-4.57, 0.4)
(-5.02, 0.356)
(-6.31, 0.413)
  };
  
  \addplot[color=green, mark=square*, line width=0.1mm]
    coordinates 
    {
(-3.35, 0.519)
(-4.39, 0.733)
(-4.93, 0.96)
(-5.02, 1)
};

\node at (-9, -2.7) {\footnotesize \textbf{more accurate}};
\draw[line width=0.4mm,-{latex[width=1mm, length=2mm]}] (-6.5,-3)--(-11.5,-3) node[right] {};

\node[rotate=270] at (-12.4, 0) {\footnotesize \textbf{faster}};
\draw[line width=0.4mm,-{latex[width=1mm, length=2mm]}] (-11.8, 1.0)--(-11.8,-1) node[right] {};

\legend{EFGP, SKI, FLAM, RLCM}
\end{axis}
\end{tikzpicture}
\caption{$1$D, squared-exponential kernel, $\ell = 0.1$}
\end{subfigure}
\begin{subfigure}{0.480\textwidth}
  \centering
\begin{tikzpicture}[scale=0.7]
\centering
\begin{axis}[
    ymin= -3, ymax=2.6,
    xmin=-14, xmax=0,
    ytick={-2, -1, 0, 1, 2},
    xlabel=$\log_{10} \textrm{EEPM}_{\textrm{new}}$, 
    xtick={-12, -10, -8, -6, -4, -2, 0}, 
    ylabel=$\log_{10} $ time (s),
    legend pos = south east,
    x dir=reverse
]

\addplot[color = red, mark=square*, line width=0.3mm]
    coordinates 
    {
(-1.42, -1.85)
(-1.62, -1.24)
(-2.3, -0.385)
(-3.26, 0.4)
(-4.26, 0.948)
  };
  
\addplot[color = blue, mark=square*, line width=0.3mm]
    coordinates 
    {
(-0.455, -0.0635)
(-1.41, 0.891)
(-1.76, 0.935)
(-2.24, 1.03)
  };
  
\addplot[mark=square*, line width=0.3mm]
    coordinates 
    {
(-13.4, 0.666)
(-13.3, 0.68)
(-13.1, 0.687)
(-13.2, 0.671)
(-13.3, 0.677)
(-13.2, 0.675)
  };

\addplot[color=green, mark=square*, line width=0.3mm]
    coordinates 
    {
(-2.07, 0.526)
(-2.21, 0.797)
(-2.32, 1.11)
  };
  
\legend{EFGP, SKI, FLAM, RLCM}
\end{axis}
\end{tikzpicture}
\caption{$1$D, Mat\'ern-1/2 kernel, $\ell = 0.1$}
\label{f1}
\end{subfigure}
\\
\vspace{0.5cm}
\begin{subfigure}{0.480\linewidth}
  \centering
\begin{tikzpicture}[scale=0.7]
\centering
\begin{axis}[
    ymin= -3, ymax=2.2,
    xmin=-9, xmax=1,
    ytick={-2, -1, 0, 1, 2},
    xlabel=$\log_{10} \textrm{EEPM}_{\textrm{new}}$, 
    xtick={-12, -10, -8, -6, -4, -2, 0}, 
    ylabel=$\log_{10} $ time (s),
    legend pos = south west,
    x dir=reverse
]

\addplot[color = red, mark=square*, line width=0.1mm]
    coordinates 
    {
(-7.85, -0.618)
(-6.59, -0.93)
(-5.78, -0.423)
(-4.79, -0.892)
(-3.88, -1.29)
(-2.86, -1.34)
(-1.87, -1.32)
  };
  
\addplot[color = blue, mark=square*, line width=0.1mm]
    coordinates 
    {
(-0.162, 1)
(-2.13, 1.36)
(-2.07, 2)
  };
  
\addplot[mark=square*, line width=0.1mm]
    coordinates 
    {
(-3.18, 0.827)
(-3.76, 0.83)
(-4.24, 0.979)
(-5.29, 1.26)
  };
  
  \addplot[color=green, mark=square*, line width=0.1mm]
    coordinates 
    {
(-1.22, 0.355)
(-1.62, 0.6)
(-1.87, 0.97)
(-2.19, 1.26)
(-2.78, 1.46)
  };
  
\legend{EFGP, SKI, FLAM, RLCM}
\end{axis}
\end{tikzpicture}
\caption{$2$D, squared-exponential kernel, $\ell = 0.1$}
\end{subfigure}
\begin{subfigure}{0.480\linewidth}
  \centering
\begin{tikzpicture}[scale=0.7]
\centering
\begin{axis}[
    ymin= -3, ymax=3.0,
    xmin=-8.1, xmax=1.1,
    ytick={-2, -1, 0, 1, 2},
    xlabel=$\log_{10} \textrm{EEPM}_{\textrm{new}}$, 
    xtick={-8, -6, -4, -2, 0}, 
    ylabel=$\log_{10} $ time (s),
    legend pos = south east,
    x dir=reverse
]

\addplot[color = red, mark=square*, line width=0.3mm]
    coordinates 
    {
(-0.951, -1.37)
(-1.01, -1.06)
(-1.37, -0.0776)
(-2.25, 1.64)
  };
  
\addplot[color = blue, mark=square*, line width=0.3mm]
    coordinates 
    {
(-0.147, 0.589)
(-0.948, 1.41)
(-1.2, 1.4)
(-1.08, 2.06)
  };
  
\addplot[mark=square*, line width=0.3mm]
    coordinates 
    {
(-3.28, 1.54)
(-4.23, 1.86)
(-5.18, 2.11)
(-6.11, 2.33)
(-6.95, 2.52)
  };

\addplot[color=green, mark=square*, line width=0.3mm]
    coordinates 
    {
(-1.41, 1.06)
(-1.51, 1.39)
(-1.54, 1.55)
  };
    
\legend{EFGP,SKI, FLAM, RLCM}
\end{axis}
\end{tikzpicture}
\caption{$2$D, Mat\'ern-1/2 kernel, $\ell = 0.1$}
\label{f8}
\end{subfigure}

\vspace{0.5cm}

\begin{subfigure}{0.480\linewidth}
  \centering
\begin{tikzpicture}[scale=0.7]
\centering
\begin{axis}[
    ytick={-2, -1, 0, 1, 2},
    xlabel=$\log_{10} \textrm{EEPM}_{\textrm{new}}$, 
    xtick={-12, -10, -8, -6, -4, -2, 0}, 
    ylabel=$\log_{10} $ time (s),
    legend pos = north east, 
    x dir=reverse
]

\addplot[color = red, mark=square*, line width=0.1mm]
    coordinates 
    {
(-7.45, 1.17)
(-6.52, 0.543)
(-5.45, 0.401)
(-4.54, 0.713)
(-3.49, 0.184)
(-2.43, 0.186)
  };
  
\addplot[color = blue, mark=square*, line width=0.1mm]
    coordinates 
    {
(-0.239, 1.91)
(-0.855, 1.95)
(-1.64, 1.97)
(-1.97, 1.99)
(-2.04, 1.98)
  };
  
\addplot[mark=square*, line width=0.1mm]
    coordinates 
    {
(-1.59, 1.67)
(-2.48, 2.09)
  };
  
  \addplot[color=green, mark=square*, line width=0.1mm]
    coordinates 
    {
(-0.646, 0.575)
(-0.704, 0.86)
(-0.79, 1.33)
  };
  
\legend{EFGP, SKI, FLAM, RLCM}
\end{axis}
\end{tikzpicture}
\caption{$3$D, squared-exponential kernel, $\ell = 0.1$}

\end{subfigure}
\begin{subfigure}{0.480\linewidth}
  \centering
\begin{tikzpicture}[scale=0.7]
\centering
\begin{axis}[
    ytick={-2, -1, 0, 1, 2, 3},
    xlabel=$\log_{10} \textrm{EEPM}_{\textrm{new}}$, 
    xtick={-4, -3, -2, -1, 0}, 
    ylabel=$\log_{10} $ time (s),
    legend pos = south east,
    x dir=reverse
]

\addplot[color = red, mark=square*, line width=0.3mm]
    coordinates 
    {
(-0.942, -0.212)
(-1.47, 1.85)
(-2.07, 3.49)
  };
  
\addplot[color = blue, mark=square*, line width=0.3mm]
    coordinates 
    {
(-0.162, 1.46)
(-0.848, 1.88)
(-0.866, 2.67)
  };
  
\addplot[mark=square*, line width=0.3mm]
    coordinates 
    {
(-1.86, 2.48)
(-2.38, 3.13)
  };

\addplot[color=green, mark=square*, line width=0.3mm]
    coordinates 
    {
(-0.94, 1.37)
(-1.01, 1.75)
(-1.08, 2.12)
  };
    
\legend{EFGP,SKI, FLAM, RLCM}
\end{axis}
\end{tikzpicture}
\caption{$3$D, Mat\'ern-1/2 kernel, $\ell = 0.1$}
\end{subfigure}
\caption{Comparisons of laptop run-time vs achieved posterior mean accuracy (EEPM$_{\textrm{new}}$) for GP regression
 in $d=1$, $2$, and $3$ dimensions, for four methods (see Section~\ref{s:comp}).
 The horizontal axis is interpreted as minus the number of correct digits.
 For each method, the curve was produced by varying a convergence parameter.
 In all cases $N=10^5$ simulated data points were used, with $\sigma = 0.5$.%
\label{fig:acc_v_time}
}
\end{figure*}
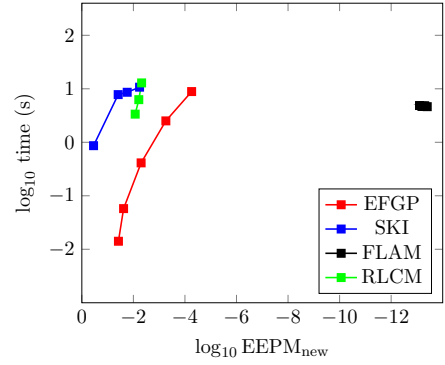
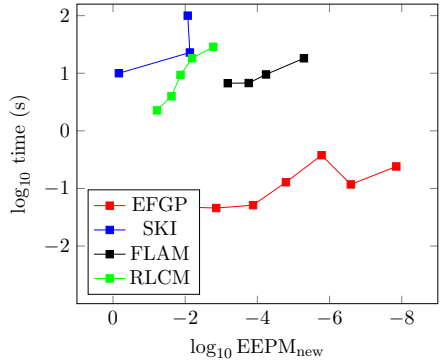
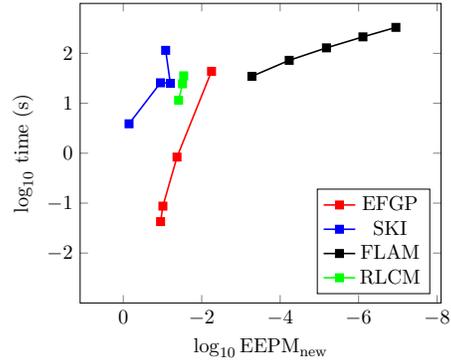
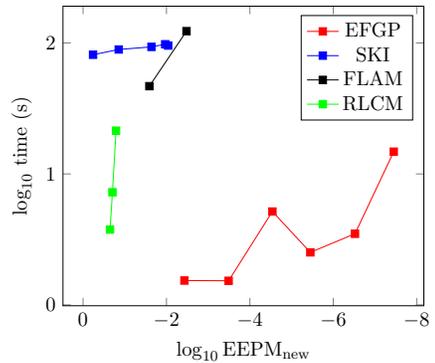
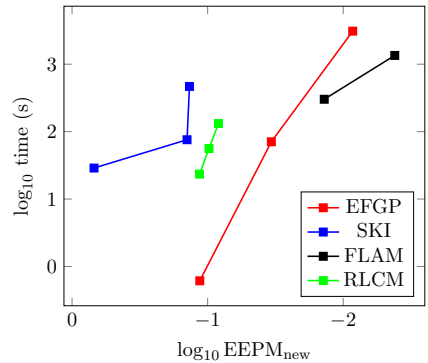

\subsection{Comparison of algorithms for GP regression}
\label{s:comp}

Here we fix $N=10^5$ in dimensions $1$, $2$, and $3$.
All aspects of the kernels
and synthetic data generation are as in the last subsection, except that
$\sigma=0.5$ is slightly larger, and the wavevectors were
$\omega=(4,3)$ in 2D and $\omega=(2,3,5)$ in 3D, up to twice as oscillatory
as before.

We benchmark EFGP against three popular state of the art methods,
chosen due to their desirable theoretical 
properties (recall~Table~\ref{t:comps}),
and efficient, user-friendly, publicly available implementations.
They represent a range of recent algorithms.
All three---indeed nearly all in the literature---are
function-space approaches, i.e., they solve the $N \times N$ linear
system in \eqref{362}. They are:
\bi
\item
  Structured Kernel Interpolation (SKI) \cite{wilson2},
  an iterative (BBMM) method using a Cartesian grid of
  inducing points with a fast apply via local cubic interpolation and the FFT.
  We wrote a MATLAB wrapper to the GPyTorch implementation \cite{gardner1}
  available at \url{https://github.com/cornellius-gp/gpytorch},
  using CPU only.\footnote{Note that GPyTorch could also exploit GPUs.}
  Accuracy was tuned via the number of inducing points per
  dimension.
\item
  Recursively Low-Rank Compressed Matrices (RLCM) \cite{chen1}, a linear-scaling fast direct solver
  implemented in C++ at
  \url{https://github.com/jiechenjiechen/RLCM}.
  We created a GitHub fork
  with additional variants of {\tt KRR\_RLCM} that can read and write
  data from binary files. These files were read/written by our MATLAB
  wrapper. We checked that the file I/O was a negligible part of the run-time.
  OPENMP was enabled, dense linear algebra used the
  OSX Accelerate framework, and we set {\tt diageps=1e-8}, {\tt refine=1},
  and {\tt par='PCA'}. Accuracy was tuned via the rank parameter.
\item
  A fast spatial GP maximum likelihood estimation method by Minden et al.\
  \cite{minden1}
  using a fast direct solver from
  ``Fast Linear Algebra in MATLAB'' (FLAM \cite{ho1}, available at
  \url{https://github.com/klho/FLAM}).
  We wrote a simple MATLAB implementation which
  factorizes $K+\sigma^2I$
  via {\tt rskelf} (recursive skeletonization), using a $d$-dimensional
  annulus containing $n_{\text{proxy}}=p^d$
  proxy points.
  A tunable accuracy parameter $\eps$ determined
  $p = \max\left(8,\lceil 0.44 \log(1/\eps) \rceil \right)$, and
  controlled the interpolative decomposition tolerance.
  We used {\tt occ=200} or {\tt 300} points per box.
\ei
See our GitHub repo (listed above)
for the precise wrappers and test codes, git submodule commits,
and configurations, allowing reproduction of the experiments.

Figure~\ref{fig:acc_v_time} shows the CPU time (vertical axis)
needed to achieve various posterior mean accuracies (EEPM$_\text{new}$,
horizontal axis), for the six possibilities
arising from the SE and Mat\'ern-$1/2$ kernels in $d=1$, $2$, and $3$.
Here the grid of test targets had
$n_{t} = 100$ points per dimension, except for $d=3$ where
$n_{t} = 30$.
The reference posterior mean $\mu(x)$
for SE kernels was
was computed using
EFGP with $\varepsilon = 10^{-14}$ in $d\le2$ and
$\varepsilon = 10^{-12}$ in $d=3$.
For Mat\'ern kernels, FLAM was used with $\eps=10^{-14}$ in $d=1$ and
$\eps=10^{-10}$ in $d=2$, while EFGP
with $\varepsilon = 5 \times 10^{-4}$ was used for $d=3$.

Figure~\ref{fig:acc_v_time} shows that for the SE kernel in all dimensions,
EFGP is the fastest method, usually by two orders of magnitude
at the same accuracy. EFGP also allow higher accuracies to be achieved
(7 of more correct digits) with minimal extra cost.
For Mat\'ern, EFGP is 1-2 orders of magnitude faster than SKI or RLCM
at the same accuracy, and is achieves around 1 extra digit of accuracy
than could be obtained with the other two codes.
For Mat\'ern in $d=1$ FLAM is ``exact'' (with rank 1, since the kernel
$e^{-|x-x'|/\ell}$ is semiseparable), so that 13-digit accuracy always results.
However, for Mat\'ern in all dimensions,
FLAM is the fastest approach for high accuracy computations while EFGP is still up to 2-3 orders of magnitude faster than FLAM for low precisions.

Note that the above results are for a
Mat\'ern parameter $\nu=1/2$, which is the {\em most challenging} for EFGP;
performance for $\nu>1/2$ is better because fewer Fourier nodes are needed.

\begin{figure*}[ht]
\begin{subfigure}{0.99 \linewidth}
  \centering
\includegraphics[height= 45mm]{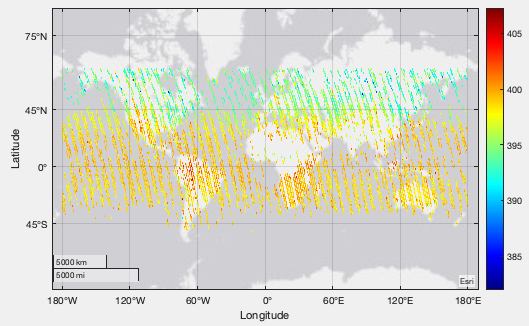}
\caption{XCO2 (ppm) raw data $\{y_n\}$ shown by color at the measurement points $\{x_n\}$.}
\label{fig:co2a}
\end{subfigure}
\par\bigskip
\begin{subfigure}{0.45 \linewidth}
  \centering
\includegraphics[height= 45mm]{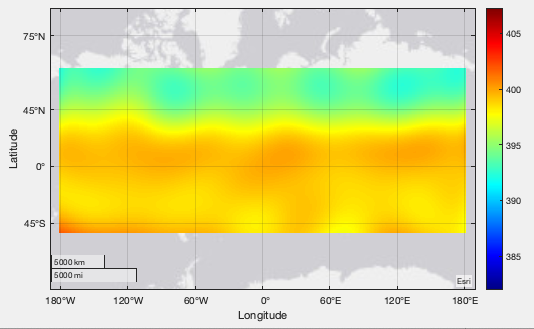}
\caption{Posterior mean $\hat\mu(x)$ for $\ell = 50$, EEPM$_{\textrm{new}}$ of $0.002$, and $0.5$ s laptop run time.}
\label{fig:co2b}
\end{subfigure}
\hspace{0.05\linewidth}
\begin{subfigure}{0.45 \linewidth}
  \centering
\includegraphics[height= 45mm]{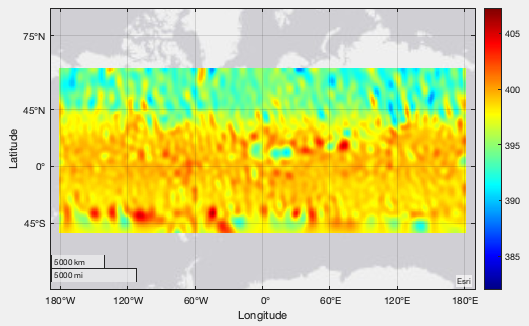}
\caption{Posterior mean $\hat\mu(x)$ for $\ell = 5$, EEPM$_{\textrm{new}}$ of $0.0002$, and $5$ s laptop run time.}
\label{fig:co2c}
\end{subfigure}

\caption{2D Gaussian process regression from
  $N \approx 1.4 \times 10^{6}$ measurement points of
  XCO2 data,
  using a squared-exponential kernel with two length scales $\ell$.
  The data was demeaned, and $\sigma^2 = 1$.
  The posterior mean $\hat\mu(x)$ is evaluated at $300\times 300$ target points.
}
\label{fig:co2}
\end{figure*}
 
\subsection{Spatial Gaussian process regression with 2D satellite data}
\label{s:co2}

In this section, we demonstrate the performance of EFGP on a 
standard large-scale GP regression in spatial statistics. 
The observed data \cite{xco2_data}
consists of $N\approx 1.4$ million 
satellite measurements of XCO2 (average
parts-per-million of $\text{CO}_2$ in an atmospheric column)
taken during the period of August 1--17, 2015,
as shown in Fig.~\ref{fig:co2}(a).
Cressie \cite{cressie2018}
describes this dataset in detail, in
addition to the data-generating process, the sources of noise,
and the importance of efficient spatial regression.
In accordance with standard practice, we use longitude-latitude coordinates 
in degrees as locations on a 2D plane. 
The observations $y_n \in (381, 408)$ are de-meaned, then
we ran GP regression to a $300 \times 300$ grid of targets,
using EFGP with the SE kernel
\begin{align}
k(x, x') \;=\; 9 e^{\frac{-|x - x'|^2 }{2\ell^2}},
\end{align}
i.e., \eqref{Gker} scaled to a prior variance of 9.
We ran two choices of length scale $\ell = 5$ and $50$,
which required total run-times of $5$ and $0.5$ seconds
respectively, giving EEPM$_{\text{new}}$ of 
$0.0002$ and $0.002$. Here the
EFGP algorithm was run with $\varepsilon = 10^{-7}$, 
and the reference posterior mean was computed using EFGP with $\varepsilon = 10^{-8}$.
For the $\sigma=1$ used,
CG required $\niter = 1916$ for $\ell=5$, and $\niter = 576$ for $\ell = 50$.
These large iteration counts are quite sensitive to the choice of
$\sigma$.
Figure \ref{fig:co2}(b--c) shows the resulting posterior means.

\begin{table}[ht]
\centering
\resizebox{\textwidth}{!}{%
\begin{tabular}{ccccccccccc}
  Alg & $\sigma$ & $\varepsilon$ & $N$ & $m$ & iters & tot (s) & mem (GB) & $\textrm{EEPM}$ & $\textrm{EEPM}_{\textrm{new}}$ & RMSE \\ 
\hline
EFGP & $0.10$ & $10^{-5}$ & $3\times 10^{6}$ & $94$ & $2853$ & $9$ & $0.1$ & $4.6\times 10^{-3}$ & $4.6\times 10^{-3}$ & $1.0\times 10^{-1}$\\ 
EFGP & $0.10$ & $10^{-7}$ & $3\times 10^{6}$ & $346$ & $9481$ & $517$ & $0.1$ & $2.0\times 10^{-4}$ & $1.9\times 10^{-4}$ & $1.0\times 10^{-1}$\\ 
FLAM & $0.10$ & $10^{-5}$ & $3\times 10^{6}$ & & & $284$ & $7.5$ & $1.9\times 10^{-3}$ & $2.9\times 10^{-2}$ & $1.0\times 10^{-1}$\\ 
FLAM & $0.10$ & $10^{-7}$ & $3\times 10^{6}$ & & & $384$ & $9.1$ & $5.4\times 10^{-5}$ & $3.0\times 10^{-4}$ & $1.0\times 10^{-1}$\\ 
\hline \hline 
EFGP & $0.10$ & $10^{-5}$ & $10^{7}$ & $94$ & $2634$ & $10$ & $0.3$ & $3.9\times 10^{-3}$ & $3.9\times 10^{-3}$ & $1.0\times 10^{-1}$\\ 
EFGP & $0.10$ & $10^{-7}$ & $10^{7}$ & $346$ & $15398$ & $878$ & $0.7$ & $3.4\times 10^{-4}$ & $3.4\times 10^{-4}$ & $1.0\times 10^{-1}$\\ 
FLAM & $0.10$ & $10^{-5}$ & $10^{7}$ & & & $732$ & $23.0$ & $6.2\times 10^{-2}$ & $5.8\times 10^{-1}$ & $5.9\times 10^{-1}$\\ 
FLAM & $0.10$ & $10^{-7}$ & $10^{7}$ & & & $1272$ & $25.0$ & $8.0\times 10^{-5}$ & $4.6\times 10^{-4}$ & $1.0\times 10^{-1}$\\ 
\hline \hline 
EFGP & $0.10$ & $10^{-5}$ & $3\times 10^{7}$ & $94$ & $1915$ & $9$ & $2.6$ & $3.1\times 10^{-3}$ & $3.1\times 10^{-3}$ & $1.0\times 10^{-1}$\\ 
EFGP & $0.10$ & $10^{-7}$ & $3\times 10^{7}$ & $346$ & $23792$ & $1315$ & $2.8$ & $5.4\times 10^{-4}$ & $5.4\times 10^{-4}$ & $1.0\times 10^{-1}$\\ 
FLAM & $0.10$ & $10^{-5}$ & $3\times 10^{7}$ & & & $1777$ & $51.9$ & $1.8\times 10^{-1}$ & $2.6\times 10^{0}$ & $2.6\times 10^{0}$\\ 
FLAM & $0.10$ & $10^{-7}$ & $3\times 10^{7}$ & & & $3328$ & $54.6$ & $1.0\times 10^{-4}$ & $7.7\times 10^{-4}$ & $1.0\times 10^{-1}$\\ 
\hline \hline 
EFGP & $0.10$ & $10^{-5}$ & $10^{8}$ & $94$ & $1393$ & $14$ & $9.3$ & $2.3\times 10^{-3}$ & $2.3\times 10^{-3}$ & $1.0\times 10^{-1}$\\ 
EFGP & $0.10$ & $10^{-7}$ & $10^{8}$ & $346$ & $35905$ & $2055$ & $9.5$ & $7.6\times 10^{-4}$ & $7.6\times 10^{-4}$ & $1.0\times 10^{-1}$\\ 
\hline \hline 
EFGP & $0.10$ & $10^{-5}$ & $10^{9}$ & $94$ & $1027$ & $103$ & $96.7$ & $1.2\times 10^{-3}$ & $1.2\times 10^{-3}$ & $1.0\times 10^{-1}$\\ 
EFGP & $0.10$ & $10^{-7}$ & $10^{9}$ & $346$ & $66199$ & $4048$ & $97.0$ & $7.9\times 10^{-4}$ & $7.9\times 10^{-4}$ & $1.0\times 10^{-1}$\\ 
\hline \hline 
\end{tabular}
}%

\caption{Timings (on the desktop), memory usage, and accuracies, for
  2D massive-scale GP regressions using the Mat\'ern kernel, for $\ell = 0.1$, $\nu=1.5$, and fixed $\sigma$.}
\label{tbig1}
\end{table}

\begin{table}[ht]
\centering
\resizebox{\textwidth}{!}{%
\begin{tabular}{ccccccccccc}
  Alg & $\sigma$ & $\varepsilon$ & $N$ & $m$ & iters & tot (s) & mem (GB) & $\textrm{EEPM}$ & $\textrm{EEPM}_{\textrm{new}}$ & RMSE \\ 
\hline
EFGP & $0.10$ & $10^{-5}$ & $3\times 10^{6}$ & $94$ & $2853$ & $9$ & $0.1$ & $4.6\times 10^{-3}$ & $4.6\times 10^{-3}$ & $1.0\times 10^{-1}$\\ 
EFGP & $0.10$ & $10^{-7}$ & $3\times 10^{6}$ & $346$ & $9481$ & $517$ & $0.1$ & $2.0\times 10^{-4}$ & $1.9\times 10^{-4}$ & $1.0\times 10^{-1}$\\ 
FLAM & $0.10$ & $10^{-5}$ & $3\times 10^{6}$ & & & $284$ & $7.5$ & $1.9\times 10^{-3}$ & $2.9\times 10^{-2}$ & $1.0\times 10^{-1}$\\ 
FLAM & $0.10$ & $10^{-7}$ & $3\times 10^{6}$ & & & $384$ & $9.1$ & $5.4\times 10^{-5}$ & $3.0\times 10^{-4}$ & $1.0\times 10^{-1}$\\ 
\hline \hline 
EFGP & $0.18$ & $10^{-5}$ & $10^{7}$ & $94$ & $2802$ & $10$ & $0.2$ & $4.6\times 10^{-3}$ & $4.6\times 10^{-3}$ & $1.8\times 10^{-1}$\\ 
EFGP & $0.18$ & $10^{-7}$ & $10^{7}$ & $346$ & $9479$ & $531$ & $0.5$ & $1.9\times 10^{-4}$ & $1.9\times 10^{-4}$ & $1.8\times 10^{-1}$\\ 
FLAM & $0.18$ & $10^{-5}$ & $10^{7}$ & & & $722$ & $23.0$ & $8.1\times 10^{-3}$ & $5.3\times 10^{-2}$ & $1.9\times 10^{-1}$\\ 
FLAM & $0.18$ & $10^{-7}$ & $10^{7}$ & & & $1206$ & $25.0$ & $5.7\times 10^{-5}$ & $3.0\times 10^{-4}$ & $1.8\times 10^{-1}$\\ 
\hline \hline 
EFGP & $0.32$ & $10^{-5}$ & $3\times 10^{7}$ & $94$ & $2791$ & $12$ & $2.5$ & $4.7\times 10^{-3}$ & $4.7\times 10^{-3}$ & $3.2\times 10^{-1}$\\ 
EFGP & $0.32$ & $10^{-7}$ & $3\times 10^{7}$ & $346$ & $9597$ & $546$ & $2.8$ & $1.9\times 10^{-4}$ & $1.9\times 10^{-4}$ & $3.2\times 10^{-1}$\\ 
FLAM & $0.32$ & $10^{-5}$ & $3\times 10^{7}$ & & & $1770$ & $51.9$ & $1.8\times 10^{-3}$ & $4.0\times 10^{-2}$ & $3.2\times 10^{-1}$\\ 
FLAM & $0.32$ & $10^{-7}$ & $3\times 10^{7}$ & & & $3298$ & $54.6$ & $5.5\times 10^{-5}$ & $3.4\times 10^{-4}$ & $3.2\times 10^{-1}$\\ 
\hline \hline 
EFGP & $0.58$ & $10^{-5}$ & $10^{8}$ & $94$ & $2603$ & $18$ & $9.2$ & $4.6\times 10^{-3}$ & $4.6\times 10^{-3}$ & $5.8\times 10^{-1}$\\ 
EFGP & $0.58$ & $10^{-7}$ & $10^{8}$ & $346$ & $9559$ & $551$ & $9.5$ & $1.9\times 10^{-4}$ & $1.9\times 10^{-4}$ & $5.8\times 10^{-1}$\\ 
\hline \hline 
EFGP & $1.83$ & $10^{-5}$ & $10^{9}$ & $94$ & $2544$ & $121$ & $96.8$ & $4.6\times 10^{-3}$ & $4.6\times 10^{-3}$ & $1.8\times 10^{0}$\\ 
EFGP & $1.83$ & $10^{-7}$ & $10^{9}$ & $346$ & $9430$ & $719$ & $97.0$ & $1.9\times 10^{-4}$ & $1.9\times 10^{-4}$ & $1.8\times 10^{0}$\\ 
\hline \hline 
\end{tabular}
}%

 \caption{Timings (on the desktop), memory usage, and accuracies, for 2D massive-scale GP regressions using the Mat\'ern kernel, for $\ell = 0.1$, $\nu=1.5$, and fixed $N/\sigma^2$.}
\label{tbig2}
\end{table}

\subsection{EFGP and FLAM performance for massive-scale 2D regression}
\label{s:big}

Based on the results in Section~\ref{s:comp}, we see that EFGP and FLAM
are the two most competitive methods of the four tested for GP regression in 2D.
In this section, we investigate their performance
on very large datasets $N \gg 10^{6}$. 
We used synthetic data as described at the start of this section,
with $\omega=(3,4)$.
We study the Mat\'ern kernel with $\nu = 3/2$ (an intermediate smoothness)
and $\ell = 0.1$.
We compare EFGP and FLAM in two different
regimes: in
Table~\ref{tbig1}, $\sigma$ (the standard deviation for both the additive noise and the regression noise model) is held fixed as $N$ grows,
while in Table~\ref{tbig2},
$N/\sigma^2$ is fixed
(a scaling expected to give constant posterior variance in the limit of
dense data).
The tables explore training data sizes $N$ up to a billion,
and two algorithm tolerances $\eps$.
As before, $m$ and $\niter$ pertain only to
EFGP and denote the number of positive Fourier discretization nodes
in each dimension, and the number of iterations of conjugate gradient.
``tot'' denotes the total run time, including evaluation of the posterior
mean $\hat\mu(x)$ on a target grid of $n_{t} = 1000$ points per dimension.
``mem (GB)'' estimates the peak RAM usage.
The reference posterior means for evaluating the errors EEPM and EEPM$_{\textrm{new}}$ were computed using EFGP with $\varepsilon = 10^{-8}$.

Focusing first on results at tolerance $\eps=10^{-7}$,
for moderate $N=3\times 10^6$, EFGP and FLAM are about the same speed,
but as $N$ grows EFGP becomes about three times faster than FLAM
despite some huge iteration counts (${>}10^4$).
Both achieve similar accuracies (3-4 digits)
and RMSE values matching $\sigma$, indicating successful regression.
However, the RAM footprint of EFGP, asymptotically linear in $N$,
is {\em much smaller} than FLAM, by factors
between 20 and 100.
This means that on the desktop EFGP can reach $N=10^9$ while
FLAM could not be run for $N>10^{8}$ points.

The story for the looser tolerance $\varepsilon = 10^{-5}$ with fixed $\sigma$ 
(Table~\ref{tbig1}) is strikingly different:
with FLAM both error metrics diverge as $N$ grows, indicating
poor regression for $N\ge 10^7$,
while EFGP retains around 3-digit accuracy and an RMSE matching $\sigma$.
Curiously, for EFGP at this looser tolerance,
$\niter$ is up to 60 times smaller than at $\eps=10^{-7}$,
and now {\em decreases} with $N$, against expectations from
condition number growth.
However, this does not seem to hurt the accuracy, with
the net result that running EFGP at this looser tolerance
is up to 100 times faster than either EFGP or FLAM at $\eps=10^{-7}$.
EFGP takes only 2 minutes for 3-digit accuracy at $N=10^9$.

Turning our attention to Table~\ref{tbig2} where $N/\sigma^2$ is held fixed, we observe that FLAM with $\varepsilon=10^{-5}$ tends to perform better than in the case where $\sigma$ itself was fixed.
In fact, now $\niter$ is quite steady for either tolerance,
consistent with an $\bigO(N / \sigma^2)$ condition number.
Thus, since the precomputation is negligible,
the EFGP timing is almost independent of $N$ for $N\le 10^{8}$.
EFGP is six times faster than FLAM for $N = 3 \times 10^{7}$ at $\varepsilon = 10^{-7}$, and $150$ times faster for $\varepsilon=10^{-5}$.

Finally, although Theorem~\ref{t:muerr} gave much more pessimistic bounds
for new targets than for the training points,
we see that for EFGP, in both tables, EEPM$_{\textrm{new}}$ is usually
very similar to EEPM.
For FLAM, in contrast, EEPM$_{\textrm{new}}$ is usually 1 digit worse than
EEPM, indicating less generalizability to new targets.

\section{Conclusions and generalizations} \label{sec:conclusion}

We have described EFGP, a fast iterative algorithm for Gaussian process
regression built upon kernel approximation by equispaced
quadrature in the Fourier domain.
The nonuniform FFT allows essentially a single pass through the data.
Not only is the resulting weight-space
linear system size $M$ independent of the data size
$N$, but its Toeplitz-type
system matrix may then be applied efficiently via a $d$-dimensional FFT.
The algorithms of this paper perform especially strongly with smoother
kernels (rapid decay in Fourier space), small dimension $d$,
and large $N$, so are expected to open up applications in 
time series, spatial, and spatio-temporal statistics.
Numerical experiments on posterior mean prediction
(i.e., kernel ridge regression) with given kernels
show large speed-ups over three state-of-the-art methods,
and handle $N$ up to a billion
even without attempting preconditioning. 
Such problems are 2--3 orders of magnitude bigger than previously reported
with a standard workstation.

We also provide several relevant analytical results,
including a uniform kernel approximation error bound, and
the numerical parameters achieving a user-requested kernel error.
Our main theorem bounds on the error in the posterior mean
relative to exact inference,
at training or at new test points, in terms of kernel approximation error
and the residual due to approximate iterative solution of the linear system.
New numerical evidence suggests that they are far from sharp.
Proving sharper estimates may demand assumptions on the data point distribution,
and is an open area of research.

This work suggests a few future other directions.
A simple one would be the evaluation of the posterior variance:
aside from the fact that a new linear solve is needed for each target, this
would be a routine application of EFGP as presented.
Another is to address the extremely large iteration counts that occur
as $N/\sigma^2$ increases. Even though the FFT renders
each iteration cheap, its cost starts to dominate for the largest problems
considered.
The weight-space linear system \eqref{be}
takes (after diagonal scaling) the form Toeplitz plus diagonal,
and preconditioners are yet to be explored.
Finally, sampling from Gaussian process priors
could be done efficiently in the equispaced Fourier basis, by
evaluating
$f(x) = \beta_1 \phi_1(x) + \dots + \beta_M \phi_M(x)$
with iid random coefficients using the
nonuniform FFT as in step 6 of Algorithm \ref{a1}.

Finally, we note that Fourier grid
methods do not scale well with dimension, because $M=\bigO(m^d)$ total modes are
required. In practice this may limit them to, say,
$d\le 6$ for smooth kernels (such as squared-exponential)
or $d\le 4$ for non-smooth (such as Mat\'ern-$\sfrac{1}{2}$).
However, the other methods compared in this work (FLAM, RLCM, and SKI)
similarly suffer, due to ranks (or number of inducing points) also
growing exponentially in $d$.
This illustrates that getting high accuracy, for high $d$, and large data
size $N$, remains extremely challenging.

\section*{Acknowledgments}
The authors are grateful for helpful discussions with
Leslie Greengard, Jeremy Hoskins, Jie Chen, Ken Ho, 
Victor Minden, Andrew Gordon Wilson, and Amit Singer.

\bibliographystyle{siamplain}
\bibliography{refs}

\end{document}